\documentclass[twoside,11pt]{article}

\usepackage{jmlr2e,amsmath,amssymb,booktabs, algorithm,natbib}
\usepackage[noend]{algpseudocode}
\usepackage[english]{babel}
\usepackage[utf8]{inputenc}
\usepackage{xcolor}

\newcommand{\trans}{^{\mbox{\tiny {\sf T}}}}

\def\E{\mathbb E}

\def\R{\mathbb R}

\usepackage{lastpage}
\jmlrheading{22}{2021}{1-\pageref{LastPage}}{10/20; Revised
9/21}{11/21}{20-1221}{Xiaowu Dai and Michael I.~Jordan}
\ShortHeadings{Learning In Decentralized Matching Markets}{Dai and Jordan}

\firstpageno{1}

\begin{document}

\title{Learning Strategies in Decentralized Matching Markets under Uncertain Preferences}

\author{\name Xiaowu Dai \email xwdai@berkeley.edu\\
       \addr 
       Department of Economics\\
       University of California, Berkeley, CA 94720-1776, USA
       \AND
       \name Michael I.~Jordan \email jordan@cs.berkeley.edu \\
       \addr 
       Division of Computer Science and Department of Statistics\\
       University of California, Berkeley, CA 94720-1776, USA
       }

\editor{Amos Storkey}

\maketitle

\begin{abstract}
We study the problem of decision-making in the setting of a scarcity of shared resources when the preferences of agents are unknown a priori and must be learned from data. Taking the two-sided matching market as a running example, we focus on the decentralized setting, where agents do not share their learned preferences with a central authority.
Our approach is based on the representation of preferences in a reproducing kernel Hilbert space, and a learning algorithm for preferences that  accounts for  uncertainty due to the competition among the agents in the market. Under regularity conditions, we show that our estimator of preferences converges at a minimax optimal rate. Given this result, we derive optimal strategies that maximize agents' expected payoffs and we calibrate the uncertain state by taking opportunity costs into account. We also derive an incentive-compatibility property and show that the outcome from the learned strategies has a stability property. Finally, we prove a fairness property that asserts that there exists no justified envy according to the learned strategies. 
\end{abstract}


\section{Introduction}

Many real-world decision-making problems can be viewed from both an economic point of view and a statistical point of view.  The economic point of view focuses on scarcity of shared resources and the need to coordinate among multiple decision-makers. Thus, decision-makers must assess preferences over outcomes and those preferences need to interact in determining an overall set of outcomes.  The statistical point of view recognizes that preferences are often not known a priori, but must be learned from data; moreover, agents' decisions are often influenced by latent state variables whose values must be inferred in order to determine a preferred outcome.  Unfortunately, it is uncommon that these two perspectives are brought together in the literature, with economic work rarely addressing the need to learn preferences from data, and statistical machine learning rarely addressing scarcity and its consequences for decision-making.  In this paper we aim to bridge this gap, studying a core microeconomic problem---two-sided matching markets---in the setting in which preferences must be learned from data.  Moreover, we focus on \emph{decentralized} matching markets, reflecting several desiderata that are common in the machine-learning literature---that agents are autonomous and private, and that scalability and avoidance of central bottlenecks is a principal concern of an overall system design.

Two-sided matching markets have played an important role in microeconomics for several decades~\citep{rothsatomayor1990}, both as a theoretical topic and as a mainstay of real-world applications. Matching markets are used to allocate indivisible ``goods" to multiple decision-making agents based on mutual compatibility as assessed via sets of preferences.  Matching markets are often organized in a decentralized way such that each agent makes their decision independently of others' decisions. Examples include college admissions, decentralized labor markets, and online dating. 
Matching markets embody a notion of scarcity in which the resources on both sides of the market are limited.  
The \emph{congestion} is a key issue in decentralized decision-making under scarcity, as participants may not be able to make enough offers and acceptances to clear the market \citep{roth1997}.  
The uncertainty of participants' preferences is ubiquitous in real-world decentralized matching markets. For instance, college admissions in the United States face applicants' uncertain preferences.  The admitted students of a college may receive offers from other colleges. Students need to accept one or reject all offers, often within a \emph{short} period.   The process provides little opportunity for the college to learn students' preferences, which depends on colleges' competition and colleges' uncertain popularity in the current year.   Consequently, the college may end up enrolling too many or too few students relative to its capacity \citep{avery2003}. It is in the colleges' interest to decide which  applicants to admit, such that the entering class will meet reasonably close to their quotas and be close to the attainable optimum in quality \citep{gale1962}.

When traditional machine learning methods are applied to problems involving scarcity and competition, they can create problems involving congestion, poor allocations, and lack of fairness \citep{busoniu2008comprehensive, finocchiaro2021bridging}.  These problems need to be addressed by mechanisms that go beyond mere load-balancing or post-processing of the outputs of a pattern recognition system.  Indeed, there are interactions between matching and learning that must be addressed jointly.  Examples include: (1) The effects of scarcity imply that preferences are non-stationary from the perspective of any individual agent. (2) Agents’ decisions are often influenced by latent state variables whose values must be inferred in practice. This presents calibration challenges that need to take into account capacity limits and opportunity costs. (3) Learning methods should be asked to produce outcomes that are not merely optimal in a pointwise sense, but stable in a game-theoretic sense.  Thus, lack of stability needs to inform the learning process.

In this work, we present a statistical framework that explicitly models a decision-making process under uncertainty and scarcity and yields learned strategies that only use local information. 
We use the classical college admissions market as our running example.  In the proposed model, there are a set of agents (for example, colleges), each with limited capacity, and a set of arms (for example, students),   each of which can be matched to at most one agent. Agents value two attributes of an arm: a ``score" (for example, SAT/ACT score) that is common to all agents and a ``fit" (for example, college-specific essay) that is agent-specific and is independent across agents.  According to their score and fit, agents rank arms, but they do not observe arms' preferences, which have no restriction.  The model incorporates the arm’s uncertain preference into an acceptance probability, depending on both the unknown state and agents’ competition. We want to learn arms’ acceptance probabilities from the historical data of arms’ binary choices, where arms’ attributes may vary over time. Various statistical learning algorithms allow efficient learning of the acceptance probability under the proposed model. To fix ideas, we present the penalized log-likelihood method in the reproducing kernel Hilbert space (RKHS) to learn  the acceptance probability. 
Under regularity conditions the estimator converges at a minimax optimal rate.
We focus on the single-stage decentralized matching that involves a single stage interaction: agents simultaneously pull sets of arms (for example, colleges offer admissions to students). Each arm accepts one of the agents (if any) that pulled it. We propose an efficient algorithm
called calibrated decentralized matching (CDM) for maximizing the agents’ expected payoff. We calibrate the unknown state by perturbing the state and balancing the marginal utility and the marginal penalty. The proposed calibration procedure takes the opportunity costs into account. The CDM can perform the calibration in both the average-case and worst-case scenarios,  depending on maximizing the averaged or minimal expected payoff concerning the unknown true state.

The CDM algorithm can serve as a recommendation engine in 
decentralized platform for decision-making.  Indeed, our results add a learning component to help participants decide which participants on the other side of the market are the best to connect to.  In particular, even with an unknown state, agents can estimate the probability of successfully pulling an arm using historical data.  The prediction of match compatibility is also possible in another direction that arms can learn how much an agent may prefer them. The CDM procedure eases congestion in practice via recommendation and pulling more arms compared to alternative methods. 
Moreover, we show that asymptotically CDM makes it safe for agents to act straightforwardly on their preferences. That is, CDM ensures incentive-compatibility for agents. We also show that CDM yields a stable outcome for the market. Our notion of stability is similar to that of  \citet{liu2014}, which is designed to study decentralized matching with incomplete information.  This stability notion extends the classical notion of stability due to \citet{gale1962} that assumes complete information of participants' preferences.  We prove that CDM is asymptotically stable. Moreover, we show that CDM is asymptotically fair for arms, in the sense of no justified envy in the matching process according to CDM.

Beyond the setting of matching markets, there are many other real-world machine learning problems that involve scarcity and competition, and in which decentralized methods of the kind that we explore here are appropriate. For example, decision-making problems in manufacturing, information gathering, and load balancing can often be treated as decentralized problems, where  time and capacity are scarce and agents need to build up beliefs about the outcome of their actions from noisy observations \citep[cf.][]{ooi1996decentralized,cogill2006approximate, seuken2008formal}. 
Related problems, such as the analysis of social dilemmas \citep{leibo2017multi} and competitive games \citep{tampuu2017multiagent} all operate in a multi-agent domain in which agents compete for scarce resources. 
Our model framework  serves as a step towards explicitly modeling a decision-making process under uncertainty and scarcity and developing  decentralized algorithms for learning from data.

\subsection{Related Work}
\label{sec:relatedwork}
We briefly review some of the related work from multiple literatures, including matching markets, mechanism design,  and multi-agent learning.

\paragraph{Matching Markets} Most theoretical work on matching markets traces back to the papers of \citet{gale1962} and \citet{shapley1971assignment}. \citet{gale1962} formulated a model of two-sided matching without side payments which they called the \emph{marriage problem} for one-to-one matching and the \emph{college admissions problem} for many-to-one matching.  \citet{shapley1971assignment} formulated a model of two-sided matching with side payments  which they called the \emph{assignment game}, which is also called the \emph{maximum weighted bipartite matching} \citep{karlin2017game}. 
Recently, there has been a surge of literature on the online version of assignment game and its generalizations, due to the important new application domain of Internet advertising \citep{mehta2013online}.
Although there exist striking analogies between some of the results for the college admissions problem and those for the assignment game \citep{rothsatomayor1990}, the key difference is that in college admissions there are agents and arms on both sides of the assignment  who care about the outcome \citep{roth1993stable}. 
This matters for a variety of reasons, not least of which is that if a matching outcome disagrees with the liking of the agents and arms involved, groups of them may be able to disregard the matching outcome by making private matches among themselves. This imposes constraints such as stability and fairness on what matching outcomes can be achieved. 
The main focus of our paper is the college admissions problem, which has no side payments and no centralized clearinghouse for coordination. Our goal is to design algorithms for college admissions that achieve the constraints such as stability and fairness under \emph{uncertain preferences}. 

\paragraph{Decentralized Matching Markets}
There has been significant work in the economics literature on congestion in decentralized markets, where participants cannot make enough offers and acceptances to clear the market.
\citet{roth1997} discussed such a decentralized market for graduating clinical psychologists, focusing on the market's timing aspect, which lasts over a day.
They found that such a decentralized but coordinated
market exhibits congestion since the interviews that a student could schedule were limited, and the resulting matching was unstable.
\citet{das2005two} considered that both sides of the decentralized market have uncertain preferences and presented an empirical study of the resulting matchings.
Unlike these works, we provide both an analytical model and optimal learning algorithms for managing congestion. \citet{haeringer2011} considered the decentralized job matching with complete information, where agents and arms are assumed to know the entire  sequence of actions employed by their opponents. 
\citet{coles2013preference} showed how introducing a signaling device in a decentralized matching market alleviates congestion.
By contrast, we consider incomplete information where agents do not know their opponents' actions, and focus on the optimal strategy under uncertain preference without signaling. 

\paragraph{Bandit Learning in Decentralized Matching Markets}
A recent thread of research focuses on matching markets with bandit learning. 
\citet{liu2019,liu2020bandit} and \citet{basu2021beyond} extended the multi-armed bandit framework to a setting in which multiple agents compete with respect to their preferred matches. They proposed an algorithm to achieve low cumulative regret for decentralized matching.
\citet{cen2021regret} further incorporated costs and transfers in this bandit learning setting to faithfully model competition between agents.
\citet{sankararaman2021dominate} studied the decentralized serial dictatorship setting, a two-sided matching market where the agents have heterogeneous valuations over the arms, while the arms have known uniform preference over agents. 
Unlike these works, we study the optimal strategy for single-stage matching, which involves no accumulated regret and allows arms to have heterogeneous preferences over agents.

\paragraph{College Admissions}
Another closely related line of work is the algorithmic literature  on college admissions  \citep[see][]{gale1962}.
The general setting involves multiple colleges competing for students in decentralized markets.
Recently, \citet{epple2006} modeled equilibrium admissions, financial aid, and enrollment.
\citet{fu2014} studied effects of tuition on equilibria, incorporating application costs, admissions, and enrollment.
Unlike these works, we emphasize students' multidimensional abilities and the uncertainty of students' acceptance probability.
\citet{chade2006} provided an optimal algorithm for students’ application decisions, characterized as a portfolio choice problem.
\citet{hafalir2018} discussed student efforts instead of colleges’ response to congestion in  decentralized college admissions with restricted applications.
We instead consider colleges' optimal strategies.
\citet{avery2010} studied early admissions when colleges have no enrollment uncertainty and showed that a cutoff strategy is optimal in equilibrium. 
\citet{chade2014} developed a decentralized Bayesian model of college admissions for two colleges and a continuum of students under a particular preference structure. 
By contrast, we study multiple colleges in the face of enrollment uncertainty.
\citet{azevedo2016} adopted a continuum model for students in a centralized market. They found a characterization of equilibria in terms of supply and demand. This model is different from the decentralized market studied in our paper. 
\citet{che2016} considered aggregated uncertainty in college admissions and focused on two colleges and a continuum of students.
By contrast, our model considers incomplete information with multiple colleges and a finite number of arms.
We study a statistical model for learning strategies in the face of enrollment uncertainty using historical data.

\paragraph{Stochastic Knapsack Problems}
Our optimization problem relates to the stochastic knapsack problems that have applications in advertising, logistics, medical diagnosis and robotics \citep{ross1989stochastic, dean2008approximating, bhalgat2011improved}.
In the stochastic knapsack problems, the rewards of items are deterministic, but the sizes of items are independent random variables with known, arbitrary distributions. The actual size of an item is known only when it is placed into the
knapsack.  The objective of a stochastic knapsack problem is to add items sequentially (one by one) into the knapsack of a given capacity  to maximize the expected reward of the items that fit into the knapsack.
Unlike the stochastic knapsack problems, our optimization problem involves the  challenge of unknown and uncertain acceptance distributions due to latent state variables. It requires new algorithms for calibrating states and finding the optimal set of arms based on the historical data.

\paragraph{Multi-Agent Learning} 
Our work also relates to multi-agent learning in noncooperative  settings.
\citet{seuken2008formal} provided an overview of algorithms that address  decentralized learning and coordination of multiple decision makers under uncertainty. See also \citet{busoniu2008comprehensive} for a survey of multi-agent learning approaches. 
Several papers have attempted to address the non-stationarity in multi-agent environments by conditioning on other agents' policies or using importance sampling \citep{tesauro2003extending,foerster2017stabilising}. 
However, these works require knowing other agents' policies or actions, a condition which is impractical in decentralized matching markets. 
By contrast, our model framework is applicable given limited information regarding other agents.  
Recently,
\citet{tampuu2017multiagent} investigated a deep $Q$-learning approach to train competing Pong agents.
\citet{foerster2018counterfactual} and \citet{lowe2017multi} studied policy gradient with actor-critic for multi-agent learning. 
Our work differs from these in that it requires less coordination and information flow between agents. 
\citet{hu2014multi} presented a prediction market model with multiple agents, showing that the market model solves the dual of a certain machine learning problem. By contrast, our model is motivated by a different class of markets---two-sided matching markets---and the model framework that we present aims at solving machine learning problems involving scarcity and uncertain preferences.

\subsection{Contributions and Outline}
\label{sec:contribution}
We develop a framework for learning uncertain preferences that drive decision-making in the context of scarcity of shared resources. We summarize our principal methodological and theoretical contributions as follows.
\begin{itemize}
\item We provide a general introduction to the problem of learning under scarcity of shared resources, as exemplified by two-sided matching markets (Section \ref{sec:model}). We develop a modeling framework for uncertain preferences that depends on a local state component (Theorem \ref{prop:accepwelldefine}).
We then construct an estimator for learning preferences under this model. We show that under regularity conditions the estimator converges at a minimax optimal rate (Theorem \ref{thm:optestpi}).
\item We prove that a version of cutoff strategy approximates an optimal strategy for maximizing an  agent’s expected payoff (Section \ref{sec:mainresults}).
In doing so, we provide near-optimal bounds on the loss of expected payoff (Theorem \ref{thm:cutoffaystate}). We propose a new data-driven method, the calibrated decentralized matching (CDM) method, to calibrate latent state variables for the purposes of determining a preferred outcome (Theorems \ref{thm:optimalstates} and \ref{cor:introofvalueV}). CDM takes the opportunity cost and penalty for exceeding the quota into account. Moreover, we discuss extensions of CDM to other real-world machine learning problems in which there is a scarcity of shared resources.
\item We provide theoretical characterizations of the decisions derived from the CDM (Section \ref{sec:propertyCDM}). Specifically, we show an incentive-compatibility property (Theorem \ref{cor:cdmagentincentive}), show the stability of the matching outcomes obtained from CDM (Theorem \ref{thm:stabcutoff}), and we show the fairness of the matching process (Theorem \ref{prof:faircdm}). 
\item We conduct simulation studies (Section \ref{sec:simulation}) to
assess the properties of CDM and compare the payoffs with alternative methods. 
We find that arms can be better off under decentralized matching compared to centralized matching. 
We demonstrate aspects of the theoretical results through experiments (Section \ref{sec:application}) using real data from college admissions and simulated graduate school admissions.  
\end{itemize}
\noindent 
We conclude the paper with further research directions in Section \ref{sec:discussion}. All proofs  are provided in the Supplementary Appendix. 


\section{Learning under Scarcity and Uncertain Preferences}
\label{sec:model}

In this section, we define a learning model under a scarcity of shared resources and uncertain preferences. Denote a set of $m$ agents  by $\mathcal P=\{P_1,P_2,\ldots,P_m\}$ and a set of $n$ arms by $\mathcal A=\{A_1, A_2, \ldots,A_n\}$, where $\mathcal P$ and $\mathcal A$ are participants on the two sides of the market. 
Each agent can attempt to pull multiple arms and there are no constraints on the overlap among the choices of different agents. 
When multiple agents select the same arm, only one agent can successfully pull the arm, with the choice of agent made according to the arm's preferences. For example, $\mathcal P$ and $\mathcal A$ might represent colleges and students in the college admissions market, or firms and workers in the decentralized labor market. 
The matching markets involve scarcity as the resources on both the agent and arm sides are limited. For example,
colleges send admission offers to applicants. When multiple colleges send offers to the same applicant, the applicant can accept at most  one offer. Each agent $P_i$ has a limited quota of successfully pulling  $q_i\geq 1$ arms,
where $q_1+q_2+\cdots+q_m\leq n$. 
We denote $[m]\equiv\{1,\ldots,m\},[n]\equiv\{1,\ldots,n\},$ and $[K]\equiv\{1,\ldots,K\}$.

In decentralized matching markets, agents and arms make their decisions independently of the decisions made by others. This feature distinguishes decentralized matching from centralized matching, which makes use of central clearinghouses to coordinate decision-making.  Notable examples of centralized matching include the national medical residency matching  \citep{roth1984} and public school choice  \citep{Abdulkadiroglu2003}.  While  centralized matching has been the major focus of the literature on matching, we see decentralized matching as  having a potentially greater range of applications, and as providing a better platform on which to bring machine-learning tools to bear.  Indeed, decentralized learning and decentralized matching are natural when (i) arms have uncertain preferences that depend on an unknown state that has a local component, and (ii) agents possess incomplete information on other agents' decisions. 

\subsection{Running Example}
Our running example is the college admissions market, where colleges match with students \citep[see][]{gale1962}. 
College admissions in countries such as the United States, Korea, and Japan are organized in a decentralized way, where colleges make admission offers to applicants. The admitted students accept or reject the offers, often within a short period of time.  Consider a college $P_i\in \mathcal P$ with a quota of enrolling  $q_i$ students. 
It is not satisfactory for college $P_i$ to only make offers to the $q_i$ best-qualified students  since some  students may reject offers. 
The enrollment uncertainty for $P_i$ can be attributed to a lack of two types of information. First, $P_i$ has little knowledge of which other colleges  admit the applicants admitted by $P_i$. Second, $P_i$ is uncertain about students' preferences: how each applicant ranks the colleges that she has applied to. Thus colleges do not know their own popularity in the current year, which is an aggregate over the uncertain rankings by each applicant. In the face of such uncertainty, college $P_i$ needs to decide which  applicants to admit such that the entering class will meet reasonably close to its quota $q_i$ and be close to the attainable optimum in quality.

\subsection{Latent Utility}

We assume that agents' deterministic preferences are a function of underlying latent utilities. In particular, we consider the following latent utility model:
\begin{equation}
\label{eqn:defofutility}
U_i(A_j) = v_j + e_{ij},\quad\forall i\in[m] \text{ and } j\in[n],
\end{equation}
where $v_j\in[0,1]$ is arm $A_j$'s \emph{systematic score}, which is available to all agents, and $e_{ij}\in[0,1]$ is an agent-specific \emph{idiosyncratic fit} available only to agent $P_i$. 
For example, in college admissions, $v_j$ can be a function of student $A_j$'s test score on a nationwide test observed by all colleges. The $e_{ij}$ corresponds to a function of student $A_j$'s performance on college-specific essays or tests conducted by college $P_i$.
In Appendix \ref{sec:anovautility}, we show that a general utility function with multidimensional scores and fits
can be transformed to the model (\ref{eqn:defofutility}) via the ANOVA decomposition. Thus, the separable structure in  model (\ref{eqn:defofutility}) is without less of generality and allows us to characterize the pattern of competition of agents in Section \ref{sec:mainresults}. 
Similar separable structures have been used in the matching markets literature \citep[cf.][]{choo2006,menzel2015, chiappori2016,  ashlagi2019}. 
The analysis that we present in this paper will also hold if one restricts the range of $e_{ij}$  to $[0,\bar{e}]$ with some $\bar{e}<1$. This restriction may be necessary for some applications. For instance, suppose the idiosyncratic fit in college admissions is viewed as more important to colleges than the systematic score. In this case, enrolled students may find it unfair to have other enrolled students with significantly lower test scores than theirs. This, in turn, would result in a reputation cost for the college.

\subsection{Uncertain Preferences}
\label{sec:stateofnature}

We now consider the uncertainty of an arm's acceptance of an agent's pulling, which depends on the arm's uncertain preference. The arm's uncertainty mainly consists of two parts. The first part is the \emph{state of the world}, which determines an arm's preference over all agents. Let the parameter $s_i\in[0,1]$ represent the local state of the world  for agent $P_i$ for $i\in[m]$ \citep{savage1972}.
The second part of the uncertainty is the agents' strategies, which is because an arm would only accept its most preferred agent among those who have pulled it. 
The main difficulty in decentralized markets is that
an agent has incomplete information, which means that the agent lacks information on their opponents' strategies. The following theorem shows that  although with incomplete information, there still exists a valid probability function that models an arm's acceptance probability. 

\begin{theorem}
\label{prop:accepwelldefine}
There exists a probability mass function $\pi_i(s_i,v_j)$ for any $i\in[m],j\in[n]$ that characterizes the probability of arm $A_j\in\mathcal A$ accepting agent $P_i$, given any fixed strategies of other agents $P_{i'},i'\neq i$.
Moreover, the function $\pi_i(s_i,v_j)$ is increasing in $s_i$, and
 the expected utility  that agent $P_i$ receives from pulling arm $A_j\in\mathcal A$ is
\begin{equation}
\label{eqn:expectutility}
\begin{aligned}
\E[\text{utility}]  = (v_j+e_{ij}) \cdot \pi_{i}(s_i, v_j), \quad i\in[m] \text{ and } j\in[n].
\end{aligned}
\end{equation}
\end{theorem}
\noindent
The arms' scores $v_j$ are public information, and agents compete for arms with better scores. This competition aspect is incorporated into the acceptance probability $\pi_{i}(s_i, v_j)$ that depends on the score $v_j$. 
Moreover, since $\pi_i(s_i, v_j)$ is increasing in $s_i$, a large value of $s_i$ corresponds to the case that the agent $P_i$ is  popular. 
On the other hand, we note that  the true state is  unknown a priori to $P_i$.
For example in  college admissions, the \emph{yield} defined as the rate at which  a college's admitted students accept the offers is a proxy for the state $s_{i}$ \citep[cf.][]{che2016}. However, the yield of the current year is unknown a priori to the college.

\begin{example}
\label{eg:two-agent}
We show the explicit form of $\pi_i(s_i,v_j),i=1,2$ for a two-agent model with agents $P_1$ and $P_2$. 
Suppose that $P_1$ pulls an arm $A_j\in\mathcal A$.
Denote by $\mu_1(s_1)$  the probability that an arm prefers $P_1$ under the state $s_1$. 
Let $\sigma_2$ be $P_2$'s strategy, which is defined by $\sigma_2(v_j,e_{2j})=\mathbf 1(P_2 \text{ pulls } A_j \text{ with attributes } \{v_j,e_{2j}\})$.
Since  $A_j$ would be pulled  by  $P_1$ with probability $1-\sigma_{2}(v_j,e_{2j})$ and  pulled by both $P_1$ and $P_2$ with probability $\sigma_{2}(v_j, e_{2j})$, $A_j$ would accept  $P_1$ with probability $1-\sigma_{2}(v_j,e_{2j}) + \mu_1(s_1)\sigma_{2}(v_j,e_{2j})$.
Since $e_{2j}$ is unknown to $A_j$ and $\sigma_2$ depends on $e_{2j}$, the expected probability that $A_j$ accepts $P_1$ is as follows,
\begin{equation*}
    \pi_1(s_1,v_j) = 1-\E[\sigma_{2}(v_j,e_{2j})] + \mu_1(s_1)\E[\sigma_{2}(v_j,e_{2j})],
\end{equation*} 
where the expectation is taken over  $e_{2j}$. 
Moreover, we can derive an explicit form of $\pi_2(s_2,v_j)$ by following a similar argument as above.
\end{example}

\subsection{Learning Arms' Uncertain Preferences}
\label{sec:kernelucb}

Agents decide  which arms to pull based on the attributes of the arms and historical matches. 
In general, we do not make the assumption that an arm with the same score and fit appears in the historical data in the matching market. Indeed, no college admissions applicant likely has exactly the same attributes as the applicants in previous years. Even the ``repeated applicants” will often change their records in their second applications. The repeated applicants either have no offer or reject all their offers and wait a full year to apply for colleges. The value of learning from the historical matches accrues when the data are used to estimate an untried arm’s acceptance probability.

Denote by $\mathcal A^t = \{A_1^t,A_2^t,\ldots,A_{n_t}^t\}$ the set of arms at  $t\in[T]\equiv\{1,\ldots,T\}$. Let $s_i^t$ be the state of agent $P_i$  at time $t$. The state $s_i^t$ is \emph{unknown} until time $t+1$, and it is non-stationary and varies over $t$. For instance,  the yield rate  of a college may change over the years.  
 For any arm $A_j^t\in\mathcal A^t$, there are an associated pair of score and fit values, $(v_j^{t}, e_{ij}^{t})$, obtained from Eq.~(\ref{eqn:defofutility}), where $ i\in[m], j\in[n_t].$
We refer to $(v_j^t,e_{ij}^t)$ as the attributes of arm $A_j^t$. 
Define the set
\begin{equation*}
\mathcal B^t_i = \{j \ | \  P_i\text{ pulls arm } A_j^t \text{ at time } t \text{ for } 1\leq j\leq n_t\},
\end{equation*}
where $|\mathcal B^t_i| = n_{it}$ and $n_{it}\leq n_t$.
For any $j\in\mathcal B^t_i$,  the outcome that $P_i$ observes is whether an arm $A_j^t$ accepted $P_i$, that is, $y_{ij}^t  = \mathbf 1\{A_j^t \text{ accepts } P_i\}$. Here, the iid outcome $y_{ij}^t$ has the likelihood $\pi_i(s_i^t,v_j^t)$ that $y_{ij}^t = 1$ and the likelihood $1-\pi_i(s_i^t,v_j^t)$ that $y_{ij}^t=0$.


We estimate the function $\pi_i$ based on historical data.  Denote  by $\mathcal D=\{(s_i^t,v_j^t,e_{ij}^t,y_{ij}^t):i\in[m]; j\in \mathcal B_{i}^t; t\in[T]\}$ the training data. 
Let the log odds ratio be
\begin{equation*}
f_i(s_i,v) = \log\left(\frac{\pi_i(s_i,v)}{1-\pi_i(s_i,v)}\right).
\end{equation*}
There exist a variety of methods in supervised learning that can efficiently estimate the log odds ratio \citep{hastie2009elements}. For concreteness, we focus on using a penalized log-likelihood method for the log odds ratio estimation in a reproducing kernel Hilbert space (RKHS) \citep{wahba1999}.  To this end, we assume that $f_i$ is from RKHS  $\mathcal H_{K_i}$ with the reproducing kernel $K_i$. Then we find $f_i \in\mathcal H_{K_i} $ to minimize
\begin{equation}
\label{eqn:minklr}
\sum_{t=1}^T\sum_{j\in \mathcal B_{i}^t}\left[-y_{ij}^tf_i(s_i^t,v^t_j)+\log\left(1+\exp\left(f_i(s_i^t,v^t_j)\right)\right)\right]+\frac{1}{2}\sum_{t=1}^Tn_{it}\lambda_i \|f_i\|_{\mathcal H_{K_i}}^2,
\end{equation}
where $\|\cdot\|_{\mathcal H_{K_i}}$ is a RKHS norm and $\lambda_i\geq 0$ is a tuning parameter. 
Consider  a tensor product structure of the RKHS $\mathcal H_{K_i}$ defined  $K_i((s_i,v),(s_i',v')) = K_i^s(s_i,s_i')K_i^v(v,v')$ based on kernels $K_i^s$ and $K_i^v$ \citep{wahbawang1995}. Assume a random feature expansion of the following form: 
\begin{equation*}
    K_i^s(s_i,s_i') = \E_{w_s}[\phi_i^s(s_i,w_s)\phi_i^s(s_i',w_s)] \text{ and } K_i^v(v,v') = \E_{w_v}[\phi_i^v(v,w_v)\phi_i^v(v',w_v)],
\end{equation*}
where $\phi_i^s(\cdot,w_s)$ and $\phi_i^v(\cdot,w_v)$ are random features  \citep{rahimi2008}.
Denote by $\{w_{s1},w_{s2},\ldots,w_{sp}\}$ and $\{w_{v1},w_{v2},\ldots,w_{vp}\}$ the sets of $p$ independent copies of $w_s$ and $w_v$, respectively. 
Write the  feature vector as $\psi_i(s_i,v)\in\R^{p}$ with the $l$th entry equal to $\frac{1}{\sqrt{p}}\phi_i^s(s_i,w_{sl})\phi_i^v(v,w_{vl})$, for $l=1,\ldots,p$.
Then $K_i((s_i,v),(s_i',v'))$ can be approximated by the product $\psi_i(s_i,v)\trans\psi_i(s_i',v')$.
Let the matrix $\Phi_i$ have rows $\psi_i(s_i^t,v_j^t)\trans$, where $j\in \mathcal B_{i}^t$ and $t\in[T]$.
By  the representer theorem in  \citet{kimeldorf1971},
the solution to Eq.~(\ref{eqn:minklr}) has the form
$\widehat{f}_i(s_i,v) = \psi_i(s_i,v)\trans\Phi_i\trans c_i$ for some vector $c_i$.  We only need to find $\theta_i= \Phi_i\trans c_i\in\R^p$ to obtain a solution to Eq.~(\ref{eqn:minklr}):
\begin{equation}
\label{eqn:predfisve}
\begin{aligned}
&\widehat{f}_i(s_i,v) = \psi_i(s_i,v)\trans\theta_i, \quad \forall i\in[m].
\end{aligned}
\end{equation}
Denote by $Y_i$ the response vector with entries $y_{ij}^t$, $j\in\mathcal B_i^t$ and $t\in[T]$. 
Applying the Newton-Raphson method to Eq.~(\ref{eqn:minklr}) 
yields the following iterative updates for vector $\theta_i$:
\begin{equation*}
\theta_i^{(\nu+1)} = \left(\Phi_i\trans W^{(\nu)}_i\Phi_i+\sum_{t=1}^Tn_{it}\lambda_i \mathbf I\right)^{-1}\Phi_i\trans W^{(\nu)}_i\left\{\Phi_i\theta^{(\nu)}+(W_i^{(\nu)})^{-1}[Y_i-\pi^{(\nu)}_i]\right\},\ \nu\geq 1.
\end{equation*}
Here, $\theta_i^{(\nu)}$ is the $\nu$th update of $\theta_i$, and $W^{(\nu)}_i=\text{diag}[\pi^{(\nu)}_i(s_i^t,v_j^t)(1-\pi^{(\nu)}_i(s_i^t,v_j^t))]_{j\in \mathcal B_{i}^t,1\leq t\leq T}$ is a weight matrix with
 $\pi_i^{(\nu)} = [1+\exp(-f_i^{(\nu)})]^{-1}$ and $f_i^{(\nu)} = \phi_i\trans\theta_i^{(\nu)}$.
The tuning parameter $\lambda_i\geq 0$  can be selected using cross validation or 
GACV \citep[see][]{wahba1999}.
\begin{theorem}
\label{thm:optestpi}
The integrated squared error of the estimate in Eq.~(\ref{eqn:predfisve}) satisfies the following inequality:
\begin{equation*}
\E[(\widehat{f}_i - f_i)^2]\leq c_f \left[T(\log T)^{-1}\right]^{-2r/(2r+1)},  \quad \forall i\in[m],
\end{equation*}
for sufficiently large $T$, where $\lambda_i\leq c_{\lambda}[T(\log T)^{-1}]^{-2r/(2r+1)}$ and $p\geq c_p[T(\log T)^{-1}]^{-2r/(2r+1)}$. Here the constants $c_f, c_\lambda,c_p>0$ are independent of $T$, and
$r\geq 1$ denotes the smoothness of kernels  such that $K_i^s(s,\cdot)$ and $K_i^v(v,\cdot)$ have squared integrable $r$th-order derivatives. Moreover, the estimate in Eq.~(\ref{eqn:predfisve}) is minimax rate-optimal.
\end{theorem}

\noindent
We make three remarks on the theorem. First,
although $f_i$ depends on both covariates $s$ and $v$, 
the optimal rate given in the theorem is very close to the minimax rate for the one-dimensional model, with the only difference in the logarithm term. 
Second, it is possible to extend the theorem to allow different orders of the smoothness of the kernels $K_i^s$ and $K_i^v$. Finally, we can predict the acceptance probability for a new arm. 
Let  $\mathcal A^{T+1} = \{A_{1},\ldots,A_{n}\}$ be the set of arms at a new time point $T+1$, where each arm  $A_j$ has  attributes obtained from Eq.~(\ref{eqn:defofutility}).
Using the estimated log odds ratio $\widehat{f}_i$ in Eq.~(\ref{eqn:predfisve}) and given the state $s_i$ at time $T+1$, we have the estimate
\begin{equation}
\label{eqn:estpiisvj}
\widehat{\pi}_i(s_i,v_j) =\left[1+\exp(-\widehat{f}_i(s_i,v_j))\right]^{-1},
\end{equation}
which serves as the prediction of the probability that $A_j$ would accept $P_i$ at time $T+1$.

\section{Optimal Strategies in Decentralized Matching}
\label{sec:mainresults}

We now study the optimal strategies in decentralized matching at a new time point $T+1$, where the matching  involves  a  single stage interaction. First, Nature draws a state denoted by $s_i^*$ for agent $P_i$ such that the arms' preferences are realized. Next,  arms simultaneously show their interests to all agents. For example, students apply to colleges. Under the assumption that students face negligible application costs, the students'  dominant strategy is submitting applications to all colleges as they do not know how colleges evaluate their academic ability or personal essays \citep{avery2010, che2016}. 
Then, agents simultaneously decide which arms to pull based on the arm's attributes. Finally, each arm accepts one of the agents (if any) that pulled it.  

Since agents simultaneously pull sets of arms once at the new time point $T+1$ and each arm accepts one of the agents that pulled it, we refer to such matching as the \emph{single-stage decentralized matching}. Each agent's goal is to learn the optimal strategies that maximize agents’ expected payoffs based on historical data at $t=1,\ldots,T$.
In the repeated game, the number of agents are the same at $t=1,\ldots, T$.  However, our model  applies to a repeated game with an arbitrary number of agents.
See the left plot of Figure \ref{fig:inexpensive} for an illustration of the single-stage decentralized matching. 

 \begin{figure}[!ht]
    \centering
    \includegraphics[width=\textwidth]{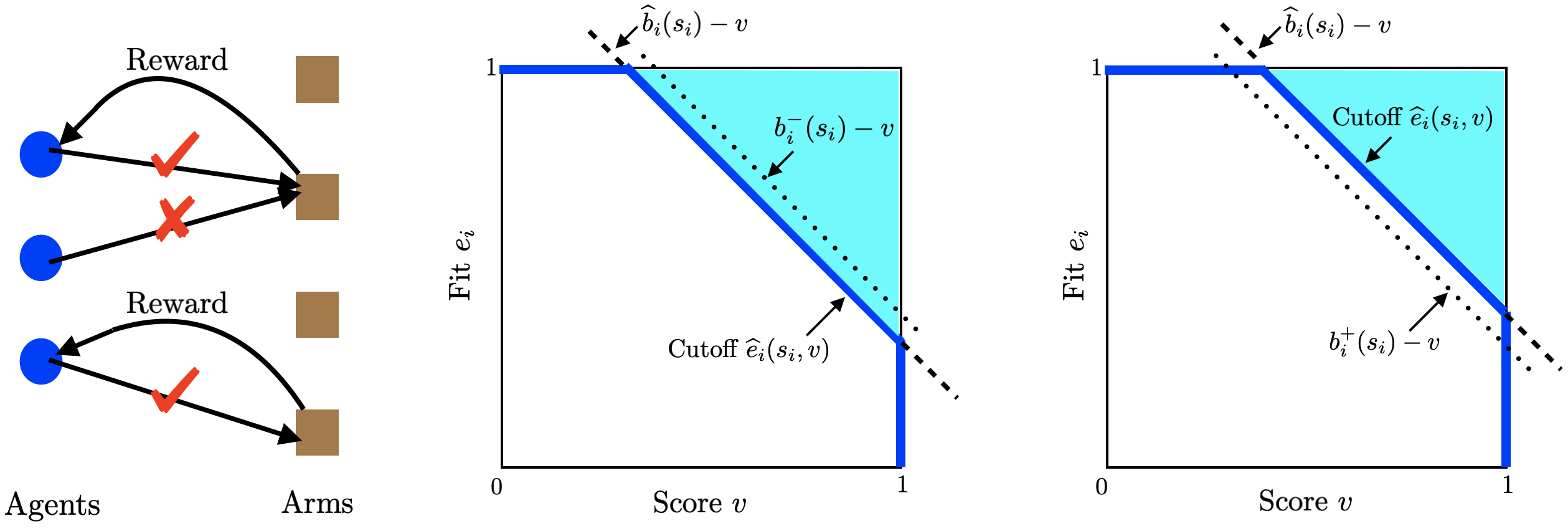}
    \caption{The left plot shows the process of the single-stage decentralized matching, with no centralized clearinghouse for coordination.
    The middle plot shows the cutoff $\widehat{e}_i(s_i,v)$ of Theorem \ref{thm:cutoffaystate} when Eq.~(\ref{eqn:jbim1+}) holds. 
    The shaded area represents $\widehat{\mathcal B}_i(s_i)$.
    The dotted line represents  $b_i^{-}(s_i)-v$ in verifying Eq.~(\ref{eqn:jbim1+}).
    The dashed  line represents $\widehat{b}_i(s_i)-v$, and if thresholding to $e_i\in[0,1]$, it yields the cutoff $\widehat{e}_i(s_i,v)$ in Eq.~(\ref{eqn:choiceofeim}) and is denoted by the  blue solid line segments.   The right plot shows the cutoff $\widehat{e}_i(s_i,v)$ when Eq.~(\ref{eqn:jbim1+}) does not hold.}
    \label{fig:inexpensive}
\end{figure}

\subsection{Algorithm}
\noindent
We propose the calibrated decentralized matching (CDM) procedure for the single-stage decentralized matching. The CDM consists of seven main steps. 
Step 1 is to predict the arms' acceptance probability. Step 2 is to estimate the distribution of the true state. Step 3 is to construct the cutoffs for a given state and its perturbation. Step 4 is to calculate the arm sets for the given state and its perturbation. Step 5 is to obtain the marginal set for the given state. In addition, Steps 3 to 5 are carried out repeatedly over a number of states. Step 6 is to calibrate the state for maximizing the agent’s average-case expected payoff, which calibration is referred to as the mean calibration. The key idea of the mean calibration is to balance the tradeoff between the opportunity cost and the penalty for exceeding the quota. 
Finally, Step 7 is to produce the final arm set under the calibrated state. 
We summarize our procedure in Algorithm \ref{alg:decentralizedcdm} first, then discuss each step in detail.

\begin{algorithm}
\caption{ \normalsize{{Calibrated decentralized matching (CDM) under the mean calibration}}}\label{alg:decentralizedcdm}
\begin{algorithmic}[1]
\State  \normalsize{\textbf{Input:} Historical data $\mathcal D=\{(s_i^t,v_j^t,e_{ij}^t,y_{ij}^t):i\in[m]; j\in \mathcal B_{i}^t; t\in[T]\}$; New arm set $\mathcal A^{T+1}$  with attributes $\{(v_j,e_{ij}):i\in[m];j\in[n]\}$ at time $T+1$; Penalty $\{\gamma_i:i\in[m]\}$ for exceeding the quota.}
\State \textbf{for} $i=1$ to $m$ \textbf{do}
\State \quad \textbf{Step 1:} Predict the acceptance probability $\widehat{\pi}_i(s_i,v_j)$ by Eq.~(\ref{eqn:estpiisvj});
\State \quad \textbf{Step 2:} Estimate the probability density function of the state $s_i^*$ by Eq.~(\ref{eqn:kdefors});
\State \quad  \textbf{for} $s_i$ in a linearly spaced vector in $[0,1]$  \textbf{do}
\State \quad\quad  \textbf{Step 3:} Construct the cutoffs $\widehat{e}_i(s_i,v)$ and $\widehat{e}_i(s_i+\delta s_i,v)$ by Eq.~(\ref{eqn:choiceofeim}) with a small $\delta s_i>0$;
\State \quad\quad  \textbf{Step 4:} Calculate the arm sets $\widehat{\mathcal B}_i(s_i)$ and $\widehat{\mathcal B}_i(s_i+\delta s_i)$ by Eq.~\eqref{eqn:armsethatbisi};
\State \quad\quad  \textbf{Step 5:} Obtain the marginal set $\partial\widehat{\mathcal B}_i(s_i)=\{\widehat{\mathcal B}_i(s_i-\delta s_i)\setminus \widehat{\mathcal B}_i(s_i)\}$;
\State \quad \textbf{end for}
\State \quad  \textbf{Step 6:} Calibrate the state $s_i$ such that the difference between LHS and RHS of Eq.~\eqref{eqn:implementmeancalibration} is below a pre-specified tolerance level;
\State \quad  \textbf{Step 7:} Calculate the arm set $\widehat{\mathcal B}_i(s_i)$ by Eq.~\eqref{eqn:armsethatbisi} under the calibrated state $s_i$.
\State \textbf{end for}
\State \textbf{Output:} The arm sets $\widehat{\mathcal B}_1(s_1), \widehat{\mathcal B}_2(s_2),\ldots, \widehat{\mathcal B}_m(s_m)$ for agents.
\end{algorithmic}
\end{algorithm}

\subsection{Agent's Expected Payoff}
\label{sec:agentsexpectedpayoff}

Since the decentralized market has no centralized clearinghouse for coordination, each agent's incentive is to act to maximize the expected payoff.
An agent's expected payoff consists of two parts: the expected utilities of arms that the agent pulls and the penalty for exceeding the quota.
Let $\mathcal B_i(s_i)\subseteq \mathcal A^{T+1}$ be the set of arms that  agent $P_i$  pulls at $T+1$ given the state $s_i$. 
We consider the following form of agent $P_i$'s expected payoff,
\begin{equation}
\label{eqn:totalutility}
\begin{aligned}
 \mathcal U_i[\mathcal B_i(s_i)] & = \sum_{j\in\mathcal B_i(s_i)}\left(v_j+e_{ij}\right)\cdot\pi_i\left(s_i,v_j\right)- \gamma_i\cdot\max\left\{\sum_{j\in\mathcal B_i(s_i)}\pi_i(s_i,v_j)-q_i,\ 0\right\}.
\end{aligned}
\end{equation}
Here, $\gamma_i>0$ denotes the marginal penalty for exceeding the quota, which originates from the scarcity of agents' capacity. The definition of $\gamma_i$ is motivated by numerous applications of decentralized matching markets. For instance, Princeton reported $1,100$ students joining its incoming class in 1995–96, significantly exceeding its capacity. As a result, Princeton had to set up mobile homes and build new dorms to accommodate the students \citep{avery2003}. 
Another example is that Korea faces the challenge of
controlling yield for colleges, since students apply to a department instead of a college, and each department faces a relatively rigid and low quota. A department that exceeds the quota is subject to a rather harsh penalty by the government in the form of sharply reduced capacities in the subsequent year \citep{che2016}. The choice of $\gamma_i$ in our framework is user-specified and it depends on the consequences of exceeding the quota. We assume that $\gamma_i>\max_{j\in\mathcal A^{T+1}}\{v_j+e_{ij}\}$; that is, the penalty is larger than an arm's latent utility.
The expected payoff $\mathcal U_i[\mathcal B_i(s_i)]$ in Eq.~(\ref{eqn:totalutility}) captures agent $P_i$'s incentive---namely, its benefit from pulling arms from the set $\mathcal B_i(s_i)$,  holding fixed its opponents' decisions.
As a result, in the optimization problem in Eq.~(\ref{eqn:totalutility}), taking the optimal action is aligned with the agent's incentive. 
Eq.~(\ref{eqn:totalutility}) excludes the situation when $P_i$ faces extra uncertainty in receiving the reward even if $P_i$ has successfully pulled an arm. For example in the dating market, Eq.~(\ref{eqn:totalutility}) models the agent's expected payoff from the dates, instead of a subsequent relationship that may eventually result from a date.

Finding and checking an optimal solution to Eq.~(\ref{eqn:totalutility}) is NP-complete; roughly because we need to individually check a significant fraction of the combinations of $\text{card}(\mathcal A^{T+1})$ arms to determine the optimal solution for Eq.~(\ref{eqn:totalutility}). Since the number of  combinations grows exponentially with the number $\text{card}(\mathcal A^{T+1})$, the complexity of any systematic algorithm becomes impractically large. 
We propose a \emph{cutoff strategy}  which greedily pulls arms according to the latent utilities and show that it is near-optimal for agents in the following theorem. 
\begin{theorem}
\label{thm:cutoffaystate}
Agent $P_i$'s cutoff strategy which pulls arm from the set
\begin{equation}
\label{eqn:armsethatbisi}
    \widehat{\mathcal B}_i(s_i) = \left.\left\{j \ \right|  \text{$A_j\in\mathcal A^{T+1}$ whose attributes $(v_j,e_{ij})$ satisfy }e_{ij}\geq \widehat{e}_i(s_i, v_j)\right\}
\end{equation}
is near-optimal as its expected payoff satisfies
\begin{equation*}
    0\leq \max_{\mathcal B_i\subseteq\mathcal A^{T+1}}\mathcal U_i[\mathcal B_i] -
    \mathcal U_i[\widehat{\mathcal B}_i(s_i)]\leq \text{UE}^\dagger.
\end{equation*}
Here the expected payoff $\mathcal U_i$ is defined in Eq.~(\ref{eqn:totalutility}), and  $\text{UE}^\dagger$ is defined by
$\text{UE}^\dagger\equiv [\min_{j\in\mathcal B_i^-(s_i)}(v_j+e_{ij})]
[q_i-\sum_{j\in\mathcal B_i^-(s_i)}\pi_i(s_i,v_j)]$.
The cutoff $\widehat{e}_i(s_i,v)$ is chosen according to Eq.~(\ref{eqn:choiceofeim}), where $\widehat{e}_i(s_i,v)$ is decreasing in $v\in[0,1]$ and satisfies $\partial\widehat{e}_i(s_i,v)/\partial v = -1$ when $\widehat{e}_i(s_i,v)\in(0,1)$.
Moreover, if there is a continuum of arms and $\pi_i(\cdot,v)$ is continuous in $v$, then $\text{UE}^\dagger=0$.
\end{theorem}
\noindent
We now specify the cutoff $\widehat{e}_i(s_i,v)$. Suppose that arms on the cutoff have the latent utility $b_i\geq 0$. The expected number of arms in $\widehat{\mathcal B}_i(s_i)$ that would accept $P_i$ is:
\begin{equation*}
\Pi_i(b_i) \equiv \sum_{j\in\mathcal A^{T+1}}\mathbf{1}\left(e_{ij}\geq \min\{\max\{b_i-v_j, 0\}, 1\}\right) \pi_i(s_i,v_j).
\end{equation*}
If there exists some $b_i\geq 0$ such that $\Pi_i(b_i) = q_i$, we let $\widehat{b}_i(s_i) = b_i$ and we have $\widehat{e}_i(s_i, v) = \min\{\max\{\widehat{b}_i(s_i)-v, 0\}, 1\}$. On the other hand, if there is no solution to  $\Pi_i(b_i) = q_i$, we let
\begin{equation*}
 b_i^{+}(s_i) = \underset{b_i\geq 0}{\arg\max}\left\{\Pi_i(b_i)> q_i\right\}\quad \text{and}\quad  b_i^{-}(s_i) = \underset{b_i\geq 0}{\arg\min}\left\{\Pi_i(b_i)< q_i\right\}.
\end{equation*}
To choose between $ b_i^{+}(s_i) $ or $b_i^{-}(s_i)$, it is necessary to balance the expected utility and the expected penalty for exceeding the quota  due to pulling arms on the \emph{boundary}.
Define two cutoffs $e_i^{+}(s_i, v) \equiv \min\{\max\{b^{+}_i(s_i)-v, 0\}, 1\}$ and $e^{-}_i(s_i, v) \equiv \min\{\max\{b^{-}_i(s_i)-v, 0\}, 1\}$,
which correspond to arm sets $\mathcal B^{+}_i(s_i) = \{j\ |\ e_{ij}\geq e_i^{+}(s_i, v_j)\}$ and $\mathcal B^{-}_i(s_i) = \{j \ | \ e_{ij}\geq e_i^{-}(s_i, v_j)\}$, respectively.
 Then the arms on the boundary are those in the discrete set $\{\mathcal B^{+}_i(s_i)\setminus
 \mathcal B^{-}_i(s_i)\}$, whose expected utility is larger than the expected penalty if
\begin{equation}
\label{eqn:jbim1+}
\sum_{j\in \mathcal B^{+}_i(s_i)\setminus \mathcal B^{-}_i(s_i)}(v_j+e_{ij})\cdot\pi_i(s_i,v_j)\geq  \gamma_i\sum_{j\in\mathcal B^{+}_i(s_i)}\pi_i(s_i,v_j)-\gamma_iq_i.
\end{equation}
Here set difference  $\{\mathcal B^{+}_i(s_i)\setminus \mathcal B^{-}_i(s_i)\}$ consists of all elements in $\mathcal B^{+}_i(s_i)$ that are not in $\mathcal B^{-}_i(s_i)$.
If Eq.~(\ref{eqn:jbim1+}) holds, 
let $\widehat{b}_i(s_i)=b_i^{+}(s_i)$; otherwise, let $\widehat{b}_i(s_i)=b_i^{-}(s_i)$. Finally, we define the cutoff
\begin{equation}
\label{eqn:choiceofeim}
\widehat{e}_i(s_i, v) = \min\left\{\max\left\{\widehat{b}_i(s_i)-v, 0\right\}, 1\right\},
\end{equation}
Here we require $\widehat{e}_i(s_i, v)\in[0,1]$ due to the assumption that $e_i$ in Eq.~(\ref{eqn:defofutility}) lies in the interval $[0,1]$.
It is clear that $\partial\widehat{e}_i(s_i, v)/\partial v=-1$ when $\widehat{e}_i(s_i, v)\in(0,1)$. Thus, $P_i$ prefers  the arms with larger latent utilities defined by Eq.~(\ref{eqn:defofutility}).
Figure \ref{fig:inexpensive} illustrates this cutoff strategy.
Although arms \emph{independently} accept or reject agents  in decentralized markets, the cutoff $\widehat{e}_i(s_i, v) $ in Eq.~(\ref{eqn:choiceofeim}) ensures that the expected number of arms accepting $P_i$ excluding those on the boundary is bounded by the quota $q_i$. According to Eq.~(\ref{eqn:jbim1+}), the arm on the boundary is pulled if the expected utility is larger than the expected penalty.

The cutoff strategy in Theorem \ref{thm:cutoffaystate} relates to the straightforward behavior \citep[see][]{fisman2006gender}, where  agents pull arms that they value more than those they do not pull.  Theorem \ref{thm:cutoffaystate} shows that although agents face uncertainty with respect to acceptance of their offers in decentralized markets, the straightforward behavior suffices.

\subsection{Calibrated Decentralized Matching (CDM)}
\label{sec:cdm}
Let $s_i^*$ be the true state for agent $P_i$ at time $T+1$.
Theorem \ref{thm:cutoffaystate} shows that the cutoff strategy maximizes $P_i$'s expected payoff if $s_i^*$ is known. However, the true state $s_i^*$ is generally unknown in practice. 
A natural question is how to calibrate the state  $s_i$ in Theorem \ref{thm:cutoffaystate}. We propose a calibration method for $s_i$ that maximizes the average-case expected payoff 
$\E_{s_i^*}\{\mathcal U_i[\widehat{\mathcal B}_i(s_i)]\}$
 over the uncertain  true state $s_i^*$. 
To formulate the theorem, we introduce some additional notation. Let $\partial \widehat{\mathcal B}_i(s_i) $  be the \emph{marginal set}, defined as
the change of the discrete set $\widehat{\mathcal B}_i(s_i)$    with respect to  a perturbation of $s_i$:
\begin{equation*}
\partial \widehat{\mathcal B}_i(s_i) \equiv \lim_{\delta s_i\to 0_+} \Big\{\widehat{\mathcal B}_i(s_i-\delta s_i)\setminus \widehat{\mathcal B}_i(s_i)\Big\}.
\end{equation*}

\begin{theorem}
\label{thm:optimalstates}
The average-case expected payoff, $\E_{s_i^*}\{\mathcal U_i[\widehat{\mathcal B}_i(s_i)]\}$, is maximized if $s_i\in(0,1)$ is  chosen as the solution to the following equation,
\begin{equation}
\label{eqn:quans}
\begin{aligned}
& \mathbb P(s_i^*\neq s_i)\sum_{j\in\partial\widehat{\mathcal B}_i(s_i)}(v_j+e_{ij})\cdot\E_{s_i^*}[\pi_i(s_i^*,v_j)\ | \ s_i^*\neq s_i] \\
&\quad\quad\quad= \gamma_i \mathbb P(s_i<s_i^*\leq 1) \sum_{j\in\partial\widehat{\mathcal B}_i(s_i)}\E_{s_i^*}[\pi_i(s_i^*,v_j)\ |\ s_i<s_i^*\leq 1].
\end{aligned}
\end{equation}
\end{theorem}

\noindent
We refer to the calibration by Theorem \ref{thm:optimalstates}
as the \emph{mean calibration} since it maximizes the agent's average-case expected payoff. 
We show in the proof that 
the average-case expected payoff has a unique maximizer that is given by the solution to Eq.~(\ref{eqn:quans}).
The key idea of the mean calibration is to balance the tradeoff between the opportunity cost and the penalty for exceeding the quota.
In particular, the left side of Eq.~(\ref{eqn:quans}) considers the opportunity cost and the marginal utility, and  the right side of Eq.~(\ref{eqn:quans}) estimates
the marginal penalty for  exceeding the quota.
We make four remarks on Theorem \ref{thm:optimalstates}.  
First, the calibrated state in Theorem  \ref{thm:optimalstates} is different from an naive estimate of the mean  $\E[s_i^*]$. 
We show in the simulation of Section \ref{sec:simulation} that the latter estimate is inefficient in maximizing the agent's expected payoff compared to the calibratin by Theorem  \ref{thm:optimalstates}. 

Second, the calibration in Theorem \ref{thm:optimalstates} takes agents' competition into account.
Note that  $\E_{s_i^*}[\pi_i(s_i^*,v_j) | s_i<s_i^*\leq 1]$  in Eq.~(\ref{eqn:quans}) is strictly increasing in $s_i$ due to the monotonicity of $\pi_i(s_i,v_j)$ in $s_i$. Thus, it is more costly for an agent to pull an arm when the agent is popular.
This result is intuitive because when an agent $P_i$ is popular, it is more likely that an arm $A_j$ pulled by multiple agents would accept $P_i$. Since $P_i$ has a larger probability of exceeding the quota when $P_i$ is popular, it is more costly for $P_i$ to pull $A_j$ compared to the case when $P_i$ is not popular.

Third, we discuss the algorithm for implementing the mean calibration in Theorem \ref{thm:optimalstates}. 
A key quantity is the probability density function of $s_i^*$, which can be estimated based on the historical data of the states $\{s_i^1,\ldots,s_i^T\}$. By the kernel density method, we obtain an  estimator for the probability density function of $s_i^*$ as follows,
\begin{equation}
\label{eqn:kdefors}
\widehat{p}_{s_i^*}(\cdot) = \frac{1}{T}\sum_{t=1}^TK_{i}^s\left(\frac{\cdot-s_i^t}{h}\right), 
\end{equation}
where $h$ is a bandwidth parameter and $K_{i}^s$ is the kernel function introduced in Section \ref{sec:kernelucb}. It is well known that the estimator $\widehat{p}_{s_i^*}(\cdot)$ achieves the minimax rate-optimality \citep{silverman1986}.
Another key quantity for implementing Theorem \ref{thm:optimalstates} is the acceptance probability $\pi_i(s_i,v)$, which can also be estimated based on the historical data as shown in Eq.~(\ref{eqn:estpiisvj}), $\widehat{\pi}_i(s_i,v) =[1+\exp(-\widehat{f}_i(s_i,v))]^{-1}$.
The last key quantity is the marginal set $\partial\widehat{\mathcal B}_i(s_i)$, which can be computed as follows. For any $s_i$ and a small perturbation $\delta s_i$, we choose the cutoffs $\widehat{e}_i(s_i,v)$ and $\widehat{e}_i(s_i+\delta s_i,v)$ according to Eq.~(\ref{eqn:choiceofeim}) and calculate the arm sets $\widehat{\mathcal B}_i(s_i)$ and $\widehat{\mathcal B}_i(s_i+\delta s_i)$ by Eq.~\eqref{eqn:armsethatbisi}, which together give the marginal set $\partial\widehat{\mathcal B}_i(s_i)=\{\widehat{\mathcal B}_i(s_i-\delta s_i)\setminus \widehat{\mathcal B}_i(s_i)\}$.
Given the estimators $\widehat{F}_{s_i^*}(s_i), \widehat{\pi}_i(s_i,v)$ and the marginal set $\partial\widehat{\mathcal B}_i(s_i)$, we calibrate $s_i$  in Theorem \ref{thm:optimalstates} as the solution to the following equation,
\begin{equation}
\label{eqn:implementmeancalibration}
\begin{aligned}
& \sum_{j\in\partial\widehat{\mathcal B}_i(s_i)}(v_j+e_{ij})\cdot\left[\int_{0}^1\widehat{\pi}_i(t,v_j)\widehat{p}_{s_i^*}(t)dt - \widehat{\pi}_i(s_i,v_j)\widehat{p}_{s_i^*}(s_i)\right]\\
&\quad\quad\quad\quad= \gamma_i\sum_{j\in\partial\widehat{\mathcal B}_i(s_i)}\left[\int_{s_i}^1\widehat{\pi}_i(t,v_j)\widehat{p}_{s_i^*}(t)dt - \widehat{\pi}_i(s_i,v_j)\widehat{p}_{s_i^*}(s_i)\right].
\end{aligned}
\end{equation}
We summarize the CDM under the mean calibration in Algorithm \ref{alg:decentralizedcdm}. It is possible to incorporate side information into  the estimation of $F_{s^*}(\cdot)$. For example, one can incorporate the belief that an agent tends to be popular at $T+1$ by overweighting the popular states in Eq.~\eqref{eqn:kdefors}.
Moreover, if Eq.~(\ref{eqn:implementmeancalibration}) has more than one solution, then $s_i$ is chosen as the largest one. 
If the probability density function of $s_i^*$ has a discrete support,  we change the objective in Eq.~\eqref{eqn:implementmeancalibration} to choosing the minimal $s_i\in[0,1]$ such that the left side of Eq.~(\ref{eqn:implementmeancalibration}) is not less than the right side of Eq.~(\ref{eqn:implementmeancalibration}). Here the search of $s_i$ starts from the maximum value in the support and decreases to the minimal value.

Finally, besides the average-case expected payoff in Theorem  \ref{thm:optimalstates}, we also consider the worst-case expected payoff concerning $s_i^*$. In the following theorem, we propose a \emph{maximin calibration}, where the calibration maximizes the minimal expected payoff $\min_{s_i^*}\{\mathcal U_i[\widehat{\mathcal B}_i(s_i)]\}$ 
over the uncertain true state $s_i^*$. 
\begin{theorem}
\label{cor:introofvalueV}
The  worst-case expected payoff  $\min_{s_i^*}\{\mathcal U_i[\widehat{\mathcal B}_i(s_i)]\}$ is maximized if $s_i\in[0,1]$ is chosen as the solution to the following equation,
\begin{equation*}
\label{eqn:jbimmaxmin}
\begin{aligned}
& \sum_{j\in\widehat{\mathcal B}_i(s_i)}(v_j+e_{ij})\cdot[\pi_i(1,v_j)-\pi_i(0,v_j)]- \gamma_i\sum_{j\in\widehat{\mathcal B}_i(s_i)}\pi_i(1,v_j)+\gamma_iq_i \\
& \quad\quad\quad\quad = \sum_{j\in\widehat{\mathcal B}_i(1)}(v_j+e_{ij})\cdot\pi_i(1,v_j)-\sum_{j\in\widehat{\mathcal B}_i(0)}(v_j+e_{ij})\cdot\pi_i(0,v_j).
\end{aligned}
\end{equation*}
\end{theorem}
\noindent
We also summarize in Algorithm \ref{alg:decentralizedcdmmaximin} of Appendix \ref{sec:CDMwithmaximin} the CDM algorithm under the maximin calibration.

The calibrations in Theorems \ref{thm:optimalstates} and  \ref{cor:introofvalueV}
take into account the opportunity cost and a penalty for exceeding the quota in order to ensure maximal payoff under scarcity. These analyses contribute to the emerging calibration literature in machine learning \citep{kennedy2001bayesian, guo2017calibration}.
Beyond decentralized matching markets, the CDM framework has applications to other noncooperative multi-agent learning problems.
For example in decision-making problems in manufacturing, information gathering, and load balancing, scarcity often arises from capacity and time constraints \citep[see][]{ooi1996decentralized,cogill2006approximate, seuken2008formal}. 
In these applications, agents have partial, heterogeneous information regarding the state of the world and the other agents.  Moreover, they need to cope with uncertainty about the outcomes arising from their actions. Related problems, such as the analysis of social dilemmas \citep{leibo2017multi} and competitive games \citep{tampuu2017multiagent} also require agents  to compete  for  scarce  resources.
The CDM framework can be extended in several ways to be suitable for these applications. 
For instance, one may be able to replace the expected payoff in Eq.~(\ref{eqn:totalutility})  with other reward functions. A primary motivation behind CDM is that defining the reward function does not require knowing the observations and actions of other agents. This is not the case for most traditional learning methods in multi-agent domains \citep{foerster2018counterfactual,lowe2017multi}. 
Moreover, it is possible to modify the single-stage payoff in Eq.~(\ref{eqn:totalutility}) to represent discounted multi-stage payoffs, a topic that we are actively investigating.

\section{Properties of CDM: Incentives, Stability and Fairness}
\label{sec:propertyCDM}

Each participant alone, on both sides of the market, knows their own preference. However, the decentralized market has no centralized system for eliciting the participants' preferences. The proposed CDM algorithm provides a mechanism to aggregate the historical data of the arms' choices. It also provides a framework for the learning of the agent strategies. This section addresses the question of whether CDM gives the agents incentives to act according to their true preferences. 
Such an incentive-compatibility property for agents is desired for the design of matching markets~\citep{roth1982}.

A related problem to incentives is the stability of the matching outcomes that we obtain. We show that any pair of agent and arm has no incentive to disregard the CDM matching and seek an alternative outcome in the context of incomplete information in decentralized markets. 

Finally, a different kind of problem concerning the matching process is fairness. For example, it is crucial to have a matching process that is fair for students in college admissions. We show that the CDM is fair for arms.

\subsection{Incentive-Compatibility}

We first need to define incentive-compatibility in decentralized matching. The payoffs of a matching problem are defined by the outcome of the matching process. 
A procedure that gives agents an incentive to act according to true preferences is defined as one that aggregates participants' actions so  that, in the resulting game, it is a dominant strategy for each agent to act on their preferences honestly. That is, agents pull arms according to the arm's latent utility in Eq.~(\ref{eqn:defofutility}).
Such a procedure is called incentive-compatible when no matter what strategies other players may play, an agent who deviates from the true preference can achieve no better outcome than if the agent had acted on their preferences honestly.

\begin{theorem}
\label{cor:cdmagentincentive}
As $T\to\infty$, the CDM gives agents an incentive to act straightforwardly on their true preferences. That is, agents pull arms according to the arm's latent utility in Eq.~(\ref{eqn:defofutility}).
\end{theorem}

\noindent
Arms pulled by multiple agents self-select based on preference. The CDM algorithm provides a mechanism to aggregate the historical matches  from arms' preference-based sorting in decentralized markets. 
Theorem \ref{cor:cdmagentincentive} affirms an incentive-compatibility  property  for  agents even with  arbitrary uncertain preference profiles of arms. This theorem also distinguishes CDM from  strategies according to arms' expected utilities in Eq.~(\ref{eqn:expectutility}) \citep[see][]{das2005two}.

We also provide a discussion on incentive-compatibility for $T<\infty$.
According to the definition, there are two parts of incentive-compatibility: (i) acting on the agents' true preferences and (ii) it is a dominant strategy in terms of maximizing the payoffs.  We discuss these two parts for the CDM for the finite $T$. First, Theorem \ref{thm:cutoffaystate} asserts that the CDM uses a cutoff strategy that greedily pulls arms according to the latent utilities in Eq.~(\ref{eqn:defofutility}).
Since the arms’ latent utilities determine the agents’ true preferences, the CDM allows the agents to act on the true preferences for the finite $T$. Second, the CDM yields the strategy that approximates the optimal solution to the optimization in Eq.~(6), which is computationally intractable and involves unknown quantities. The proposed CDM learns the unknown quantities from historical data with the finite $T$,  and achieves the minimax convergence rate, shown in Theorem \ref{thm:optestpi}. Moreover, the CDM calibrates the uncertain state by taking opportunity costs into account. Together, the CDM yields the strategy that well approximates the optimal strategy for the finite $T$. The numerical examples in Section \ref{sec:agentbetteroff} demonstrate that no matter what strategies other players may play, the alternative methods that deviate from the true preferences would achieve worse outcomes than the CDM that allows agents to act on their preferences honestly. In summary, for the finite $T$, the CDM guarantees approximate incentive-compatibility in that agents can act on their true preferences and achieve the almost optimal payoffs, which are better than alternative methods  that deviate from the true preferences.

\subsection{Stability}
\label{sec:stabilitycdm}

An agent-arm pair $(P_i,A_j)$ is referred to as \emph{matched} if $P_i$ pulls $A_j$, and $A_j$ accepts $P_i$.  A \emph{blocking pair} is an agent-arm pair  that is not  matched, but  both the agent and arm prefer to be matched together.
\citet{gale1962} defined a matching of arms to agents to be \emph{stable} if there is no blocking pair.
This definition was motivated by the concern that following a matching process, some agent-arm pairs with both the incentive and the power may deviate together,
thus hindering the implementation of the intended  matching outcome. This form of stability is widely considered to be a key factor in a successful implementation of a centralized clearinghouse \citep{roth1990}.

\citet{gale1962}  also proved that the set of stable matching is never empty.  Among all stable matchings, the one that is the most preferred by all arms is called an \emph{arm-optimal matching}. The one that is the most preferred by all agents is called an \emph{agent-optimal matching}. The arm- (or agent-) optimal matching is also agents' (or arms') least-preferred matching among all stable matchings \citep{knuth1997}.

The stability of a centralized matching relies on the stringent assumption of complete information on the preferences of the participants. In our decentralized matching formulation, it is more natural to assume that agents only have incomplete information on preferences. Thus we require a modification of the stability notion in \citet{gale1962}.
In particular, it is necessary to specify the uninformed agents' beliefs that might block a candidate stable matching \citep[cf.][]{liu2014}.
We formulate the following notion of \emph{individual rationality} to specify the beliefs of uninformed agents.
Given any $P_i$'s strategy,
let $\mathcal N_i$ be the expected number of arms that would accept $P_i$. 
An additional arm $A_j$ with attributes $(v_j,e_{ij})$, which has a smaller latent utility than those currently pulled by $P_i$, is  acceptable to $P_i$ if and only if individual rationality holds:
 \begin{equation}
\label{eqn:accepcond}
(v_j+e_{ij})\cdot\pi_i(s_i,v_j) \geq \gamma_i\cdot\max\left\{\mathcal N_i+\pi_i(s_i,v_j)-q_i,0\right\}.
\end{equation}
In other words, under $P_i$'s current strategy, $A_j$ is acceptable to $P_i$ if  and only if $A_j$'s expected utility is at least the expected penalty for exceeding the quota. 
By Eq.~(\ref{eqn:jbim1+}), the proposed CDM satisfies the individual rationality condition~(\ref{eqn:accepcond}).

A matching in a decentralized market is defined as stable  if there is no blocking pair under the individual rationality condition~(\ref{eqn:accepcond}). This stability notion is motivated by considering feasibility constraints (for example, single-stage interaction) in  decentralized markets. 
In practice, participants are often compelled to accept the outcomes in such markets. For instance, the procedure in which some high school athletes are matched to colleges involves signing ``letters of intent," which prohibits athletes from further negotiating with other colleges. 
Suppose that an arm $A_j$ does not satisfy Eq.~(\ref{eqn:accepcond}). Then $P_i$ finds that $A_j$ only has a small latent utility but pulling $A_j$ may incur a large penalty due to exceeding the quota. 
It is individually rational that $P_i$ does not pull $A_j$ in order to receive at least zero payoff.  After the matching process, suppose $P_i$ finds that $A_j$ is acceptable. That is, $P_i$ has remaining capacity, and $A_j$ prefers to match to $P_i$. Still, $P_i$ has no incentive to disregard the matching and seek an alternative outcome due to the feasibility constraints in decentralized markets. 

\begin{theorem}
\label{thm:stabcutoff}
As $T\to\infty$, the CDM procedure yields a stable matching. That is, agent and arm has no incentive to disregard the CDM matching and seek an alternative outcome in decentralized markets.
\end{theorem}
  
\noindent
There exists a lattice structure of stable matchings yielded by the CDM. 
If CDM allows each agent to pull all arms, the outcome is the arm-optimal stable matching. 
If CDM allows each agent to pull arms up to their quota, and those pulled arms are distinct, the outcome is the agent-optimal stable matching \citep{roth2008}. In practice, the stable matching yielded by CDM is generally between the agent-optimal and arm-optimal matchings due to the competition of agents and a lack of coordination. Indeed, agents are incentivized to pull more arms than their quotas in decentralized markets in order to hedge against the uncertainty of whether the pulled arms will accept them.
On the other hand, arms may benefit from this uncertainty. 
For example, suppose the college admissions market uses a clearinghouse to form a centralized market. All students apply to top colleges. In that case, the deferred acceptance algorithm yields the college-optimal matching \citep{gale1962, knuth1997}. Compared to the CDM outcomes, students are worse off in the centralized market.

In contrast to CDM, the strategies of pulling arms according to the expected utilities in
Eq.~(\ref{eqn:expectutility})
may result in an unstable matching.
For instance, if the agent $P_i$ decides according to $A_j$'s expected utility, $\pi_i(s_i,v_j)(v_j+e_{ij})$, then $P_i$ will not pull $A_j$ if  the acceptance probability $\pi_i(s_i,v_j)$ is small enough such that the expected utility of $A_j$ is the smallest among all arms, while the latent utility $(v_j+e_{ij})$ of $A_j$ is the largest. As a result, $(P_i,A_j)$ will form a blocking pair if $A_j$ is unmatched, and the matching is unstable.

We also provide a discussions on stability for $T<\infty$.
There are two cases that an agent-arm pair $(P_i,A_j)$ forms a blocking pair: (i) $P_i$ prefers $A_j$ to some of its matched arms and (ii) $P_i$ has unfilled quota, and $A_j$ is unmatched.   
We discuss these two cases of a blocking pair for the CDM for the finite $T$. 
For the first case, we note that Theorem \ref{thm:cutoffaystate} asserts that $P_i$ must have pulled $A_j$ according to the cutoff strategy of CDM.
However, $P_i$ must have been  subsequently rejected by $A_j$ in favor of some agent that $A_j$ liked better. 
Hence, $A_j$ must prefer its currently matched agent to $P_i$, and there is no instability.  This holds for finite $T$. 
For the second case, we note that $P_i$ did not pull $A_j$ since otherwise, $A_j$ would have accepted $P_i$. This happens when  the individual rationality in Eq.~(13) suggests that  $P_i$ would find $A_j$ unacceptable. 
For finite $T$, since the estimate of the population acceptance probability could be  imperfect, the second case is possible to happen. However, these neglected arms are on the boundary of acceptance. 
In summary, for finite $T$, the CDM achieves the approximate stability  in that there is no agent-arm pair such that the agent prefers the arm to some of its matched arms, and the instability occurs only when an agent missed some arms on  the boundary of acceptance and those arms end up unmatched. However, we argue that in practice, the agent has limited incentive to deviate the intended matching outcome for those arms on the boundary of acceptance. Moreover, the numerical examples in Sections 5.1 and 5.2 demonstrate that for finite $T$, the CDM is likely to pull the arms are on the boundary of acceptance, and there is no instability.

\subsection{Fairness}

We consider fairness in terms of ``no justified envy'' \citep{balinski1999,Abdulkadiroglu2003}. 
An arm $A_j$ has justified envy if there exists an agent $P_i$ such that $(P_i,A_j)$ form a blocking pair.
A matching process in decentralized markets is  fair if no arm has justified envy.

Eliminating justified envy is mathematically the same as stability in college admissions under a centralized clearinghouse \citep{gale1962}.
However, there exists a significant difference between stability and eliminating justified envy in decentralized markets. For example, each school is a strategic agent in college admissions, so stability has strategic implications, whereas eliminating justified envy is motivated by fairness considerations.
Another difference is that stability in our setting of decentralized markets requires the additional individual rationality condition~(\ref{eqn:accepcond}).

\begin{theorem}
\label{prof:faircdm}
As $T\to\infty$, the CDM is fair for arms. That is, no arm has justified envy in the matching process according to CDM. 
\end{theorem}

\noindent
We now compare CDM with the oracle arm set, where the latter maximizes agents' average-case expected payoff by assuming  complete information on opponent agents' strategies. 
For any arm set $\mathcal B_i\subseteq \mathcal A^{T+1}$, let $O_{\mathcal B_i}$ be the set of true states under which agent $P_i$'s strategy of pulling $\mathcal B_i$ results in the over-enrollment:
\begin{equation*}
O_{\mathcal B_i}\equiv\left\{s_i^*\in[0,1]\ \left|\ \sum_{j\in\mathcal B_i}\pi_i(s_i^*,v_j)>q_i\right.\right\},\quad\forall i\in[m].
\end{equation*}
Under the true state $s_i^*$, agent $P_i$'s expected payoff of pulling $\mathcal B_i$ is
\begin{equation*}
\begin{aligned}
 \mathcal U_i[\mathcal B_i] & = \sum_{j\in\mathcal B_i}\left(v_j+e_{ij}\right)\cdot\pi_i\left(s^*_i,v_j\right)- \gamma_i\cdot\max\left\{\sum_{j\in\mathcal B_i}\pi_i(s^*_i,v_j)-q_i,\ 0\right\}.
\end{aligned}
\end{equation*}
Then $P_i$'s average-case  expected payoff $\E_{s_i^*}\{\mathcal U_i[\mathcal B_i]\}$ is
\begin{equation*}
\begin{aligned}
\sum_{j\in\mathcal B_i}\left(v_j+e_{ij}\right)\cdot\E_{s_i^*}[\pi_i(s_i^*,v_j)]-\gamma_i\cdot\mathbb P\left(s_i^*\in O_{\mathcal B_i}\right)\left\{\sum_{j\in\mathcal B_i}\E_{s_i^*}[\pi_i(s_i^*,v_j)\ |\ s_i^*\in O_{\mathcal B_i}]-q_i\right\}.
\end{aligned}
\end{equation*}
Hence, the oracle arm set $\mathcal B_i^*$ for maximizing $\E_{s_i^*}\{\mathcal U_i[\mathcal B_i]\}$ becomes
\begin{equation}
\label{eqn:oracleset}
\mathcal B_i^* = \left\{j\in\mathcal A^{T+1}\ \left|\ v_j+e_{ij}\geq\gamma_i\cdot\mathbb P(s_i^*\in O_{\mathcal B^*_i})\frac{\E_{s_i^*}[\pi_i(s_i^*,v_j)\ |\ s_i^*\in O_{\mathcal B^*_i}]}{\E_{s_i^*}[\pi_i(s_i^*,v_j)]}\right.\right\}.
\end{equation}
\citet{che2016} discussed a particular case of the oracle arm set in a two-agent model. Eq.~(\ref{eqn:oracleset}) generalizes the oracle set to multiple agents. 
Although $\mathcal B_i^*$ has a closed-form expression in Eq.~(\ref{eqn:oracleset}), 
$\mathcal B_i^*$  is not estimable from training data. This is because  the probability $\mathbb P(s_i^*\in O_{\mathcal B^*_i})$ requires the knowledge of opponent agents' strategies, which are unknown in decentralized markets. 
Moreover, we show that the matching process according to the strategy of $P_i$ pulling arms from $\mathcal B_i^*$ is unfair for arms.

\begin{theorem}
\label{thm:unfairnessoracle}
The strategy corresponding to the oracle set in Eq.~(\ref{eqn:oracleset}) is unfair if for at least one agent $P_i$, there exists an interval $(v',v'')\subset[0,1]$ such that
\begin{equation}
\label{eqn:changeparpisvj}
\E\left[\left.\frac{\partial\pi_i(s_i^*,v)}{\partial v}\ \right|\ s_i^*\not\in O_{\mathcal B^*_i}\right] < \E\left[\left.\frac{\partial\pi_i(s_i^*,v)}{\partial v}\ \right|\ s_i^*\in O_{\mathcal B^*_i}\right]<0,\quad\forall v\in(v',v'').
\end{equation}
\end{theorem}
Here, Eq.~(\ref{eqn:changeparpisvj}) means that the probability of  arms with score $v\in(v',v'')$  accepting $P_i$ decreases less when $P_i$ 
is popular (i.e., $s_i^*\in O_{\mathcal B^*_i}$) than  that when $P_i$ is not popular (i.e., $s_i^*\not\in O_{\mathcal B^*_i}$). This condition holds for many decentralized matching market models, including  the two-agent model of Example \ref{eg:two-agent} in Section \ref{sec:stateofnature}.

\begin{example}
Consider the two-agent model of Example \ref{eg:two-agent}. Suppose that agent $P_2$ adopts a cutoff strategy $\sigma_2(v_j,e_{2j}) = \mathbf 1(e_{2j}\geq \widetilde{e}_2(s_1^*,v_j))$, where  $\widetilde{e}_2(s_1^*,v)\in[0,1]$ is a cutoff curve under the true state $s_1^*$. Assume that $\partial\widetilde{e}_2(s_1^*,v_j)/\partial v <0$ for $v_j\in(v',v'')$, which means that the requirement on the fit $e_{2j}$ decreases as the score $v_j\in(v',v'')$ increases. This  is a  natural assumption in college admissions and holds for many cutoffs, including the cutoff  given by Eq.~\eqref{eqn:choiceofeim}.
If $e_{2j}$ follows a uniform distribution on $[0,1]$, then
$\E[\sigma_2(v_j,e_{2j})]=1-\widetilde{e}_2(s_1^*,v_j)$.
Recall that Example \ref{eg:two-agent} gives  that 
$\pi_1(s^*_1,v_j) = 1-\E[\sigma_{2}(v_j,e_{2j})] + \mu_1(s^*_1)\E[\sigma_{2}(v_j,e_{2j})]$.
In the current example, it is equivalent to
\begin{equation*}
    \pi_1(s^*_1,v_j) =  \widetilde{e}_2(s_1^*,v_j) + \mu_1(s_1^*)(1-\widetilde{e}_2(s_1^*,v_j)).
\end{equation*}
Since $O_{\mathcal B^*_1}$ is the set of true states corresponding to the over-enrollments,  \begin{equation*}
    \E[\mu_1(s_1^*)|s_1^*\not\in O_{\mathcal B^*_1}] < \E[\mu_1(s_1^*)|s_1^*\in O_{\mathcal B^*_1}].
\end{equation*}
Together with $\partial\widetilde{e}_2(s_1^*,v_j)/\partial v <0$ for $v_j\in(v',v'')$, we have that for any $v_j\in(v',v'')$,
\begin{equation*}
\E\left[\frac{\partial\widetilde{e}_2(s_1^*,v_j)}{\partial v} (1-\mu_1(s_1^*))\left|s_1^*\not\in O_{\mathcal B_1}\right.\right]<\E\left[\frac{\partial\widetilde{e}_2(s_1^*,v_j)}{\partial v}(1-\mu_1(s_1^*))\left|s_1^*\in O_{\mathcal B_1}\right.\right]<0,
\end{equation*}
which directly implies Eq.~\eqref{eqn:changeparpisvj} for agent $P_1$.
\end{example}

 \begin{figure}[!ht]
    \centering
    \includegraphics[width=0.3\textwidth]{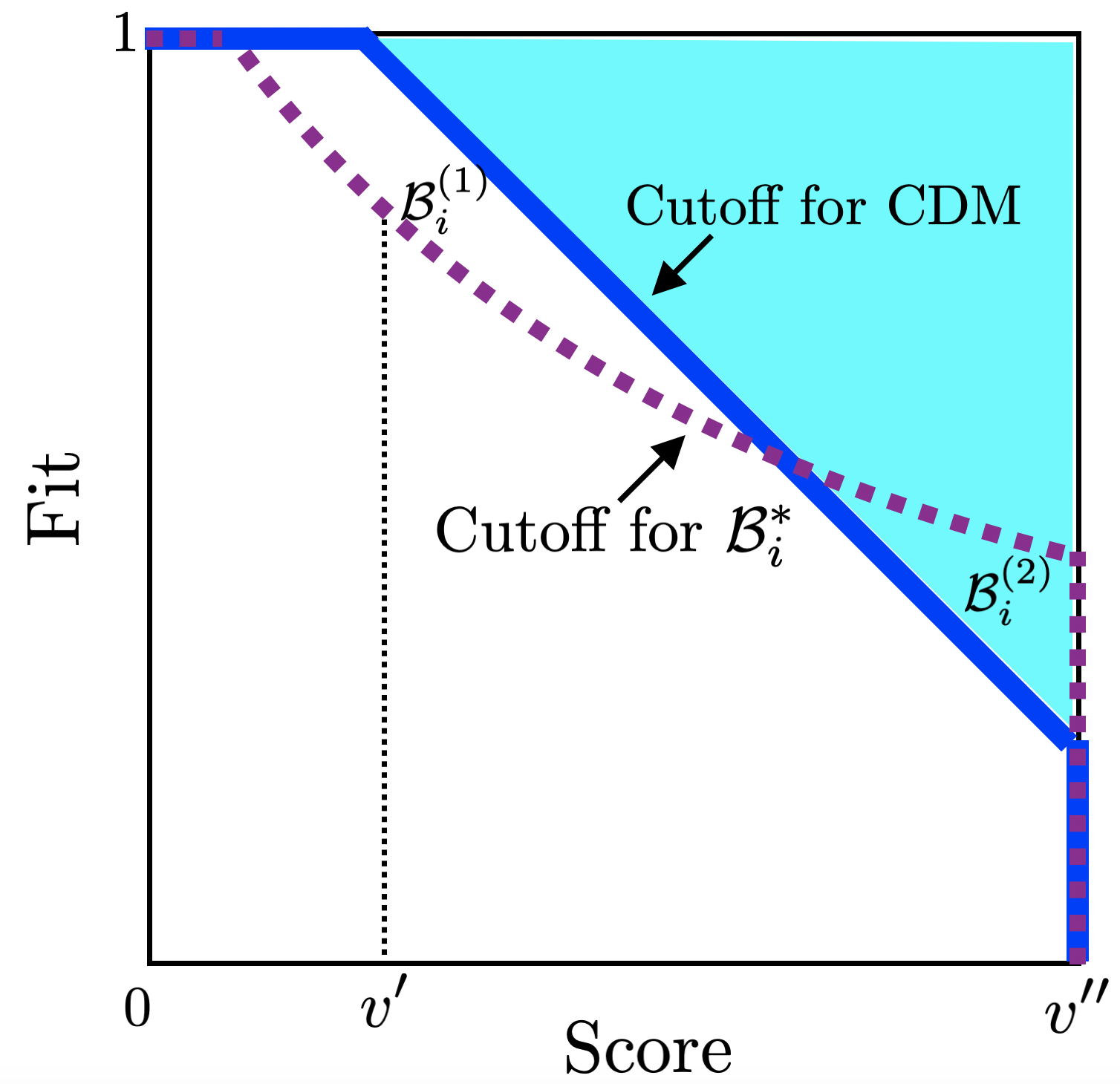}
    \caption{Fairness of CDM  compared to the unfairness of oracle arm set $\mathcal B_i^* $. The dotted curve represents the cutoff for $\mathcal B_i^*$ in Eq.~(\ref{eqn:oracleset}). The solid line segments denote the cutoff for CDM. The arms in $\mathcal B_i^{(2)}$ have justified envy towards arms in $\mathcal B_i^{(1)}$. }
    \label{fig:faircdm}
\end{figure}

Figure \ref{fig:faircdm} illustrates the unfairness of the strategy corresponding to the oracle set.
The proof of Theorem \ref{thm:unfairnessoracle} shows that the slope of $\mathcal B_i^* $'s cutoff curve is in the interval $(-1,0)$, for any $v\in(v',v'')$.
Hence, there are arms not selected to $\mathcal B_i^* $ (for example, those in $\mathcal B_i^{(2)}$). However, they rank higher than some selected arms  (for example, those  in $\mathcal B_i^{(1)}$) according to agent $P_i$'s true preference. Thus, arms in $\mathcal B_i^{(2)}$ have justified envy towards arms in  $\mathcal B_i^{(1)}$. 
On the contrary, Theorem \ref{thm:cutoffaystate} shows that the slope of CDM's cutoff equals $-1$ and hence the agent prefers arms with larger latent utilities. 

The intuition behind the results on fairness is as follows. Using the oracle arm set in Eq.~\eqref{eqn:oracleset},  agents would respond to congestion strategically to avoid head-on competition under the uncertain state. 
Specifically, agents strategically favor arms that rank highly in fits against those that rank highly in scores, for such arms are likely to be overlooked by competitors and their acceptance rates are more stable under the uncertain state. Strategic targeting makes the matching process unfair, in the sense of creating justified envy as illustrated in Figure \ref{fig:faircdm}.
In contrast, using the proposed CDM, agents would firstly calibrate the uncertain state. Under a calibrated state, arms' acceptance rates are predictable, and agents can greedily pull arms according to their true preferences. This straightforward behavior makes the matching process fair.

\section{Empirical Study}
\label{sec:simulation}

In this  section, we provide a numerical investigation of the fairness and stability properties of CDM. We also study the payoffs achieved by CDM compared to alternative methods.
\subsection{Stability and Fairness of CDM}
\label{subsec:fairstability}
Suppose there are three agents, $P_1,P_2,P_3$, $m=3$, and three arms, $A_1,A_2,A_3$, $n=3$. The latent utilities and the arms' true preferences are given in Table \ref{table:eg1}. Arms have scores and fits as follows: $v_1=v_2=v_3=2$, and $e_{11}=e_{21}=e_{32}=0,e_{13}=e_{22}= e_{31}=0.5, e_{12}=e_{23}=e_{33}=1$, respectively. Each agent has quota $q=1$ and the penalty for exceeding the quota is $\gamma=10$.
Agents have to make decisions on which arms to pull without knowing the arms' true preferences. We compare the CDM procedure with the \emph{greedy action}, which chooses arms with maximum expected utilities under the total expected acceptance up to the quota  
\citep{das2005two}.

\begin{table}[ht]
\centering
\begin{tabular}{ ccc }   
 \textbf{\small{(a) Latent utility}}  & \quad\quad\quad  \textbf{\small{(b) Arm's preference}} & \quad\quad\quad  \textbf{\small{(c) Expected utility}} \\[0.5ex]
\begin{tabular}{ c  c c c } 
\toprule
 & $A_1$ & $A_2$ & $A_3$ \\ [0.3ex]  
 \cmidrule(lr){2-2}\cmidrule(lr){3-3}\cmidrule(lr){4-4}
$P_1$ & 2 & 3 & 2.5 \\ [0.3ex]  
$P_2$ & 2 & 2.5 & 3 \\ [0.3ex]  
$P_3$ & 2.5 & 2 & 3 \\ [0.3ex]  
\bottomrule
\end{tabular} &  
\quad\quad\quad
\begin{tabular}{ c  c c c } 
\toprule
 & $P_1$ & $P_2$ & $P_3$ \\ [0.3ex]  
 \cmidrule(lr){2-2}\cmidrule(lr){3-3}\cmidrule(lr){4-4}
$A_1$ & 3 & 2 & 1 \\ [0.3ex]  
$A_2$ & 2 & 3 & 1 \\ [0.3ex]  
$A_3$ & 1 & 3 & 2 \\ [0.3ex]  
\bottomrule
\end{tabular} &  
\quad\quad\quad
\begin{tabular}{ c  c c c } 
\toprule
 & $A_1$ & $A_2$ & $A_3$ \\ [0.3ex]  
 \cmidrule(lr){2-2}\cmidrule(lr){3-3}\cmidrule(lr){4-4}
$P_1$ & $0.52$ & $1.99$ & $2.5$ \\ [0.3ex]  
$P_2$ & $0.67$ & $0$ & $0$ \\ [0.3ex]  
$P_3$ & 2.5 & $2$ & $1.05$ \\ [0.3ex]  
\bottomrule
\end{tabular} 
\end{tabular}
\caption{(a) Arm's latent utilities for each agent, which corresponds to Eq.~(\ref{eqn:defofutility}). For example, $P_1$ receives utility $2.5$ when it successfully pulls $A_3$. 
(b) Arms' preferences with the number indicating the arms' ranking of agents. For example, $A_1$ ranks $P_3$ first, $P_2$ second, $P_1$ third. These preferences are unknown to agents.
(c) Expected utilities, which corresponds to Eq.~(\ref{eqn:expectutility}). For example, $P_1$ expects to receive utility $0.52$ if it pulls $A_1$.}
\label{table:eg1}
\end{table}

\begin{figure}[!ht]
    \centering
    \includegraphics[width=\textwidth]{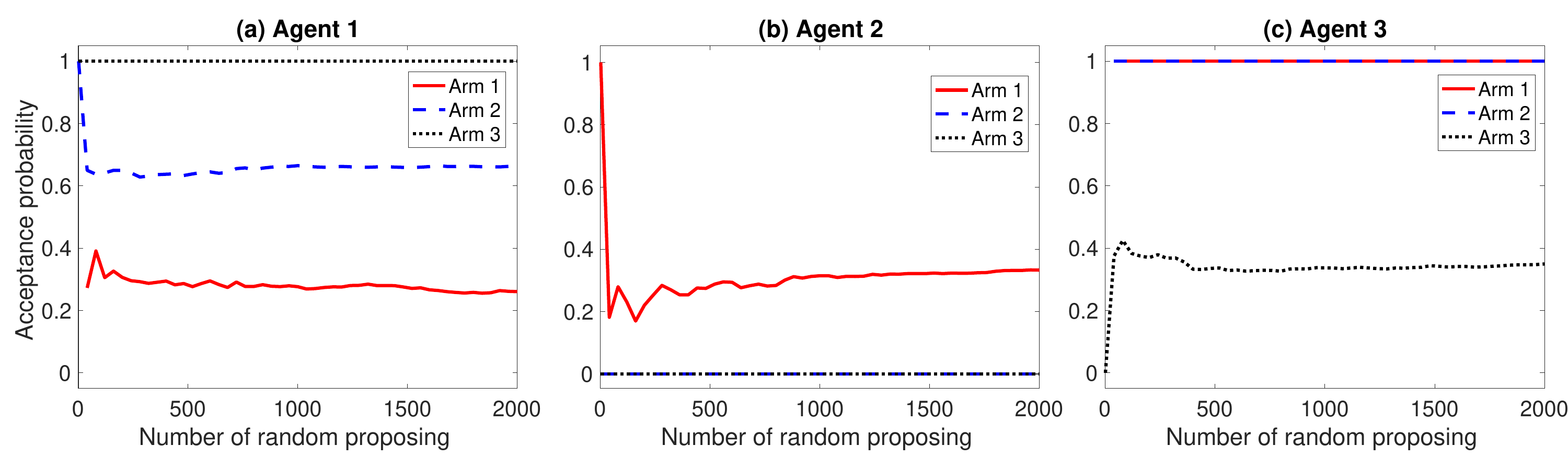}
    \caption{Arms' acceptance probabilities for three agents. (a) $P_1$, (b) $P_2$, (c)  $P_3$.}
    \label{fig:eg1}
\end{figure}

The training data are simulated by having each agent pull a random number of arms according to its latent utilities. Figure \ref{fig:eg1} shows arms' acceptance probabilities $\pi_i(s_i^*,v_j)$ based on a total of $2000$ rounds of random proposals. That is, the sample size is $T=2000$. There is a unique true state. 
We implement CDM according to Algorithm \ref{alg:decentralizedcdm}.
The  CDM procedure suggests that $P_1$ pulls $A_2$, $P_2$ pulls $A_1,A_2,A_3$, and $P_3$ pulls $A_3$, and it gives the stable matching $(A_1,P_2),(A_2,P_1), (A_3,P_3)$.
On the contrary, the greedy action suggests that $P_1$ pulls $A_3$, $P_2$ pulls $A_1,A_2,A_3$, and $P_3$ pulls $A_1$, and it yields the following matching: $(A_1,P_3),(A_2,P_2), (A_3,P_1)$. 
We note three differences between the two matchings. First, the greedy action is \emph{unfair} since  $A_2$ has justified envy towards $A_3$ in the sense that $A_2$  prefers $P_1$ to $P_2$. However, $P_1$ pulls $A_3$ that ranks below $A_2$ according to the true preference of $P_1$.  Second, the greedy action also yields an \emph{unstable} matching  since $(A_2,P_1)$ is a blocking pair. 
On the other hand,  CDM  is fair and yields a stable matching. 
The stable matching is both  agent-optimal and arm-optimal in this example.
Finally, the total payoff that agents receive using CDM equals  $3+2+3=8$, which is larger than the total payoff that agents receive from the greedy action, $2.5+2.5+2.5=7.5$.

\subsection{Lattice Structure for the Stability of CDM}
\label{sec:latstrcdm}

We consider the decentralized matching with four different preference structures: S1, S2, S3, S4. The market consists of three agents, $m=3$, and three arms, $n=3$. 
Each agent has the same quota $q=1$ and penalty $\gamma=5$. 
Table \ref{table:eg2} gives arms' latent utilities and true preferences. 
The training data are simulated by having agents pull random numbers of arms according to their latent utilities. 
The last column of Table \ref{table:eg2} gives the estimates of arms'  acceptance probabilities $\pi_i(s_i^*,v_j)$ after 2000 rounds of random proposals under each of the structures S1---S4. That is, the sample size is $T=2000$ for each of the structures S1---S4. These acceptance probabilities are evaluated at convergence.  

\begin{table}[ht]
\centering
\begin{tabular}{ ccc }   
 \textbf{\small{Latent utility of S1}}  & \quad  \textbf{\small{Arm's preference of S1}} & \quad \textbf{\small{Acceptance probability of S1}} \\[0.5ex]
\begin{tabular}{ c  c c c } 
\toprule
 & $A_1$ & $A_2$ & $A_3$ \\ [0.3ex]  
 \cmidrule(lr){2-2}\cmidrule(lr){3-3}\cmidrule(lr){4-4}
$P_1$ & 2.5 & 3 & 2 \\ [0.3ex]  
$P_2$ & 3 & 2.5 & 2 \\ [0.3ex]  
$P_3$ & 3 & 2.5 & 2 \\ [0.3ex]  
\bottomrule
\end{tabular} &  
\quad
\begin{tabular}{ c  c c c } 
\toprule
 & $P_1$ & $P_2$ & $P_3$ \\ [0.3ex]  
 \cmidrule(lr){2-2}\cmidrule(lr){3-3}\cmidrule(lr){4-4}
$A_1$ & 1 & 2 & 3 \\ [0.3ex]  
$A_2$ & 2 & 3 & 1 \\ [0.3ex]  
$A_3$ & 1 & 2 & 3 \\ [0.3ex]  
\bottomrule
\end{tabular} &  
\quad
\begin{tabular}{ c  c c c } 
\toprule
 & $A_1$ & $A_2$ & $A_3$ \\ [0.3ex]  
 \cmidrule(lr){2-2}\cmidrule(lr){3-3}\cmidrule(lr){4-4}
$P_1$ & $1$ & $0.34$ & $1$ \\ [0.3ex]  
$P_2$ & $0.35$ & $0$ & $0.65$ \\ [0.3ex]  
$P_3$ & $0$ & $1$ & $0.45$ \\ [0.3ex]  
\bottomrule
\end{tabular} 
\vspace{0.2in}
\\ 
 \textbf{\small{Latent utility of S2}}  & \quad  \textbf{\small{Arm's preference of S2}} & \quad  \textbf{\small{Acceptance probability of S2}} \\[0.5ex]
\begin{tabular}{ c  c c c } 
\toprule
 & $A_1$ & $A_2$ & $A_3$ \\ [0.3ex]  
 \cmidrule(lr){2-2}\cmidrule(lr){3-3}\cmidrule(lr){4-4}
$P_1$ & 3 & 2.5 & 2 \\ [0.3ex]  
$P_2$ & 2.5 & 3 & 2 \\ [0.3ex]  
$P_3$ & 2.5 & 2 & 3 \\ [0.3ex]  
\bottomrule
\end{tabular} &  
\quad
\begin{tabular}{ c  c c c } 
\toprule
 & $P_1$ & $P_2$ & $P_3$ \\ [0.3ex]  
 \cmidrule(lr){2-2}\cmidrule(lr){3-3}\cmidrule(lr){4-4}
$A_1$ & 3 & 1 & 2 \\ [0.3ex]  
$A_2$ & 1 & 2 & 3 \\ [0.3ex]  
$A_3$ & 2 & 3 & 1 \\ [0.3ex]  
\bottomrule
\end{tabular} &  
\quad
\begin{tabular}{ c  c c c } 
\toprule
 & $A_1$ & $A_2$ & $A_3$ \\ [0.3ex]  
 \cmidrule(lr){2-2}\cmidrule(lr){3-3}\cmidrule(lr){4-4}
$P_1$ & $0.10$ & $1$ & $0$ \\ [0.3ex]  
$P_2$ & $1$ & $0.35$ & $0$ \\ [0.3ex]  
$P_3$ & $0.32$ & $0$ & $1$ \\ [0.3ex]  
\bottomrule
\end{tabular} 
\vspace{0.2in}
\\ 
 \textbf{\small{Latent utility of S3}}  & \quad  \textbf{\small{Arm's preference of S3}} & \quad  \textbf{\small{Acceptance probability of S3}} \\[0.5ex]
\begin{tabular}{ c  c c c } 
\toprule
 & $A_1$ & $A_2$ & $A_3$ \\ [0.3ex]  
 \cmidrule(lr){2-2}\cmidrule(lr){3-3}\cmidrule(lr){4-4}
$P_1$ & 2 & 3 & 2.5 \\ [0.3ex]  
$P_2$ & 2.5 & 2 & 3 \\ [0.3ex]  
$P_3$ & 3 & 2.5 & 2 \\ [0.3ex]  
\bottomrule
\end{tabular} &  
\quad
\begin{tabular}{ c  c c c } 
\toprule
 & $P_1$ & $P_2$ & $P_3$ \\ [0.3ex]  
 \cmidrule(lr){2-2}\cmidrule(lr){3-3}\cmidrule(lr){4-4}
$A_1$ & 1 & 2 & 3 \\ [0.3ex]  
$A_2$ & 3 & 1 & 2 \\ [0.3ex]  
$A_3$ & 2 & 3 & 1 \\ [0.3ex]  
\bottomrule
\end{tabular} &  
\quad
\begin{tabular}{ c  c c c } 
\toprule
 & $A_1$ & $A_2$ & $A_3$ \\ [0.3ex]  
 \cmidrule(lr){2-2}\cmidrule(lr){3-3}\cmidrule(lr){4-4}
$P_1$ & $1$ & $0.22$ & $0.67$ \\ [0.3ex]  
$P_2$ & $0.66$ & $1$ & $0.24$ \\ [0.3ex]  
$P_3$ & $0.22$ & $0.66$ & $1$ \\ [0.3ex]  
\bottomrule
\end{tabular} 
\vspace{0.2in}
\\ 
 \textbf{\small{Latent utility of S4}}  & \quad  \textbf{\small{Arm's preference of S4}} & \quad  \textbf{\small{Acceptance probability of S4}} \\[0.5ex]
\begin{tabular}{ c  c c c } 
\toprule
 & $A_1$ & $A_2$ & $A_3$ \\ [0.3ex]  
 \cmidrule(lr){2-2}\cmidrule(lr){3-3}\cmidrule(lr){4-4}
$P_1$ & 3 & 2 & 2.5 \\ [0.3ex]  
$P_2$ & 2 & 2.5 & 3 \\ [0.3ex]  
$P_3$ & 2 & 2.5 & 3 \\ [0.3ex]  
\bottomrule
\end{tabular} &  
\quad
\begin{tabular}{ c  c c c } 
\toprule
 & $P_1$ & $P_2$ & $P_3$ \\ [0.3ex]  
 \cmidrule(lr){2-2}\cmidrule(lr){3-3}\cmidrule(lr){4-4}
$A_1$ & 3 & 2 & 1 \\ [0.3ex]  
$A_2$ & 2 & 1 & 3 \\ [0.3ex]  
$A_3$ & 1 & 3 & 2 \\ [0.3ex]  
\bottomrule
\end{tabular} &  
\quad
\begin{tabular}{ c  c c c } 
\toprule
 & $A_1$ & $A_2$ & $A_3$ \\ [0.3ex]  
 \cmidrule(lr){2-2}\cmidrule(lr){3-3}\cmidrule(lr){4-4}
$P_1$ & $0.43$ & $0.29$ & $1$ \\ [0.3ex]  
$P_2$ & $0.69$ & $1$ & $0$ \\ [0.3ex]  
$P_3$ & $1$ & $0.22$ & $0.35$ \\ [0.3ex]  
\bottomrule
\end{tabular} 
\end{tabular}
\caption{The left column shows arms' latent utilities for each agent. 
The middle column shows arms' preferences, where the number indicates the arms' ranking of agents. 
The right column shows arms' acceptance probabilities.}
\label{table:eg2}
\end{table}

We find that CDM in Algorithm \ref{alg:decentralizedcdm} gives stable matchings under all of the structures S1---S4. However, these stable matchings have different optimality in terms of agents' and arms' welfare. In particular, the matching in S1 is $(A_1,P_1),(A_2,P_3),(A_3,P_2)$, which is both agent-optimal and arm-optimal; the matching in S2 is $(A_1,P_2),(A_2,P_1),(A_3,P_3)$, which  is arm-optimal but  not agent-optimal; the matching in S3 is $(A_1,P_2), (A_2,P_3),$ $(A_3,P_1)$, which is not agent-optimal or arm-optimal;  the matching in S4 is $(A_1,P_1)$, $(A_2,P_2)$, $(A_3,P_1)$,  which is not agent-optimal or arm-optimal. This lattice structure corroborates the results in Section \ref{sec:stabilitycdm}.  We further make three remarks on these matchings. First, some arms benefit from the decentralized matching compared with the centralized matching using the agent-proposing DA algorithm. 
For example, the CDM matching in S2 is arm-optimal. The agent-proposing DA would give the agent-optimal matching, that is, the arms’ least-preferred stable matching.

Second, no strategy in decentralized markets guarantees  yielding the agent-optimal matching, mainly due to the competition of agents and a lack of coordination. Consider the S3 structure as an example, where CDM suggests that $P_1$ pulls $A_1,A_2$, $P_2$ pulls $A_2,A_3$ and $P_3$ pulls $A_1,A_3$. 
In this example, the strategy, according to the CDM, is a subgame perfect equilibrium since 
each agent's action is the best response against other agents' actions. For instance, if $P_1$ changes to the strategy by only pulling $A_2$ while other agents' strategies do not change, this gives the following matching: $(A_1,P_2),$ $(A_2,P_3),$ $(A_3,P_2)$, where $P_1$ is worse off due to the unfilled quota. If $P_1$ changes to pull all arms $A_1,A_2,A_3$ while other agents' strategies do not change, then the resulting matching is $(A_1,P_1)$, $(A_2,P_3),$ $(A_3,P_1)$, where $P_1$ is also worse off due to exceeding the quota. A similar argument applies to $P_2$ and $P_3$. 
Hence, no agent  has an incentive to change its strategy. In contrast, the agent-proposing DA
permits each agent to pull one arm each time according to its latent utility. Such coordination results in the agent-optimal matching. However, if there is no coordination, the  competition of agents generates uncertainty on the acceptance of the pulled arms. Hence, agents have the incentive to pull more arms than their quotas to combat the uncertainty.

Third, we discuss individual rationality in Eq.~(\ref{eqn:accepcond}).
Without this condition, any stable matching can be supported by a subgame perfect equilibrium, and only stable matchings can arise in equilibrium  \citep{alcalde2000}. However, under Eq.~(\ref{eqn:accepcond}), the set of stable matchings is enlarged. It includes (but
may not coincide with) the set of subgame perfect equilibria.
Consider S4 as an example, where CDM suggests that $P_1$ pulls $A_1,A_3$, and $P_2$ pulls $A_2,A_3$, and $P_3$ pulls $A_2,A_3$. 
The strategy given by CDM is not a subgame perfect equilibrium. 
However, the resulting matching is stable under Eq.~(\ref{eqn:accepcond}). For instance, $P_3$ finds $A_1$ having a larger expected penalty for exceeding the quota  than its expected utility. That is, $P_3$ finds $A_1$ unacceptable. Hence, there is no blocking pair.

\subsection{Agents' Payoffs}
\label{sec:agentbetteroff}

Consider ten agents $m=10$ and varying numbers of arms, $n=\{50,70,90,110,130,150\}$. Each arm has a score $v_j$ and fits $e_{ij}$ drawn uniformly from $[0,1]$.
The agents' preferences are determined by arms' latent utilities as in Eq.~(\ref{eqn:defofutility}).
Each agent has the same quota $q=5$ and the same penalty $\gamma$ chosen from $\{2,2.5,3\}$. 
The simulation generates random arm preferences with 10 different states from $\{s_1,\ldots,s_{10}\}\subset [0,1]$. The training data are simulated by having agents pull random numbers of arms according to their latent utilities. The training data consists of $20$ rounds of random proposals under each of the arms' preference structures. In total, the sample size is $T=200$.
The testing data draws a random state from $\{s_1,\ldots,s_{10}\}$ and generates the corresponding arms' preferences. This example simulates top colleges competing for top students. Students' preferences are uncertain and depend on colleges'  reputation and popularity in the current year. 

\begin{figure}[!ht]
\centering
\includegraphics[width=\textwidth, height=2.25in]{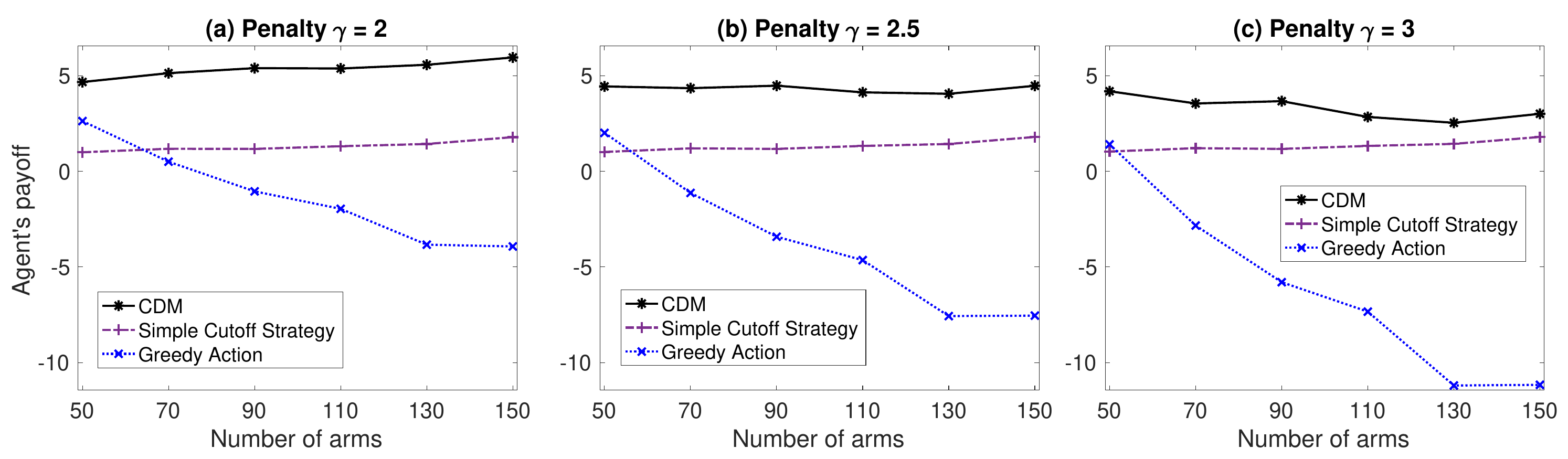}
\caption{Performance of three strategies with varying numbers of arms. The results are averaged over $500$ data replications. Penalty levels (a) $\gamma=2$, (b) $\gamma=2.5$, (c) $\gamma=3$.}
\label{fig:eg3_1}
\end{figure}

We compare the agent's payoff achieved by CDM in Algorithm \ref{alg:decentralizedcdm} with that by other methods. In particular, we  consider two alternative methods: (i) the \emph{simple cutoff strategy} where the agent chooses $q$ best arms according to the latent utility;  (ii) the \emph{greedy action} where the agent chooses arms with maximum expected  utilities and a total expected acceptance up to $q$.
Figure \ref{fig:eg3_1} reports the agent $P_1$'s averaged payoffs over $500$ data replications. Here, all  agents other than $P_1$ use the CDM with mean calibration and  $P_1$ uses one of the three methods: CDM, simple cutoff strategy, and greedy action. It is seen that CDM gives the largest average payoffs compared to alternative methods, and the advantage of CDM is robust to different numbers of arms and penalty levels. We  make two further remarks. First, the state's calibration is useful in improving the agent's utility under uncertain preferences of the arms. For example, CDM outperforms the simple cutoff strategy that has no calibration on the uncertain state.
Moreover, CDM helps to ease congestion as it pulls more arms compared to the simple cutoff strategy. 
Second, CDM performs particularly well if the matching market has  intense competition. In this simulated market where arms' preferences are random,  a smaller number of arms corresponds to a higher  competition level. 
It is seen that CDM is significantly better than other methods in the regime of  small numbers of students. 
On the other hand, the simple cutoff strategy does not work well with small numbers of students. The reason is that under an intense competition among the agents, arms reject most offers sent according to the simple cutoff strategy.

\begin{figure}[!ht]
\centering
\includegraphics[width=\textwidth, height=2.25in]{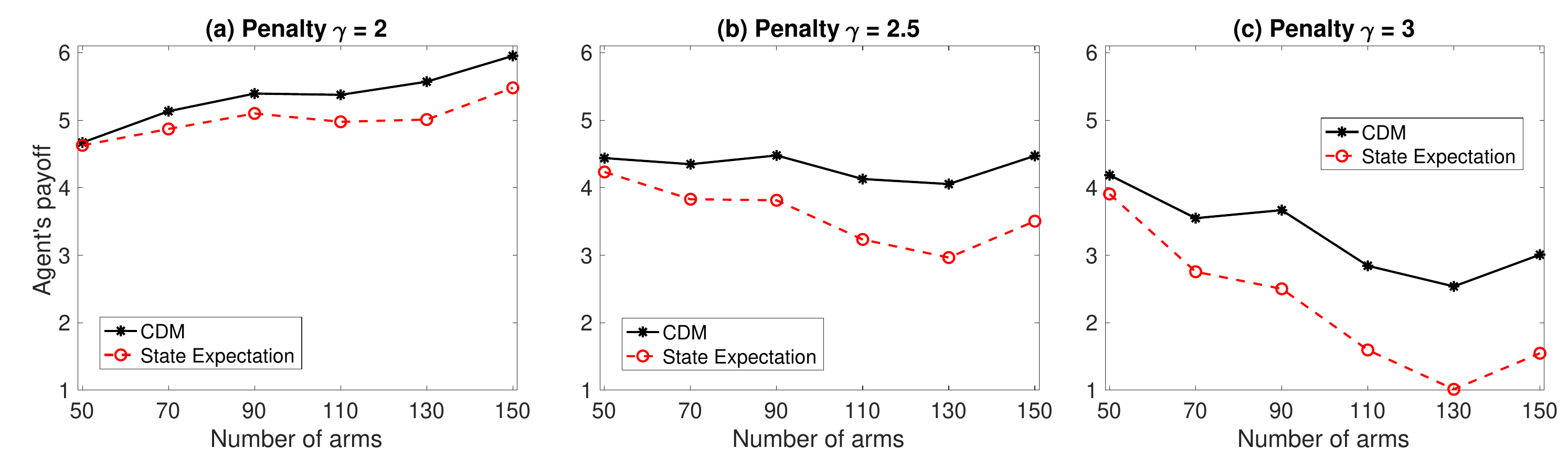}
\caption{Performance of two calibration methods  with varying numbers of arms. The results are averaged over $500$ data replications. Penalty levels (a)  $\gamma=2$, (b)  $\gamma=2.5$, (c) $\gamma=3$.}
\label{fig:eg3_2}
\end{figure}

We also compare two different calibration methods: CDM and \emph{state expectation}. The latter calibrates the unknown state using the naive mean estimate of states, $\E[s_i^*]$, which is discussed in Section \ref{sec:cdm}. 
Figure \ref{fig:eg3_2} shows $P_1$'s averaged payoffs over $500$ data replications. 
Here, all agents other than $P_1$ use the CDM with mean calibration and $P_1$ uses one of the two methods: CDM and state expectation. We observe that CDM is adaptive to different levels of penalty $\gamma$.  However, the state expectation method degrades quickly as $\gamma$ increases. This difference between the two methods is because CDM balances the marginal payoff and the marginal penalty for exceeding the quota. Hence CDM is sensitive to the penalty $\gamma$, while the state expectation method is not.


\section{Application to the Study of College Admissions}
\label{sec:application}
We demonstrate aspects of the theoretical predictions through a real data from college admissions and a simulated example on graduate school admissions.

\subsection{U.S. College Admissions}

In this section we present an analysis of admissions data from the \emph{New York Times} ``The Choice" blog (\url{https://thechoice.blogs.nytimes.com/category/admissions-data}). The data set contains admission yields in 2015-17 for $37$ U.S.\ colleges, including liberal arts colleges, national universities, and other undergraduate programs. 
That is, $m=37$ and $T=3$.
As we discussed in Section \ref{sec:stateofnature}, a college's yield is a proxy for the state $s_{i}$ as an indicator of the college's popularity.
We excluded two colleges, Harvard and Yale, from the sample due to a significant proportion of missing values. 
We test if the yields of regular admissions changed over  2015-17. The null hypothesis is that the state  is the same over 2015-17. We use a simultaneous chi-squared test for all colleges 
under false-discovery rate (FDR) control  at the $.05$ significance level \citep{benjamini1995}. 
There are $13$ colleges tested with significant $p$-values.
Figure \ref{fig:nytdata} shows that  colleges with large numbers of admitted students are more likely to have significantly varied yields compared to colleges with small numbers of admitted students.  Moreover,  top-ranked national universities and liberal arts colleges are likely to have significantly varied yields.
This result corroborates the students' uncertain preferences in college admissions.

\begin{figure}[!ht]
    \centering
    \includegraphics[width=\textwidth,height=2.5in]{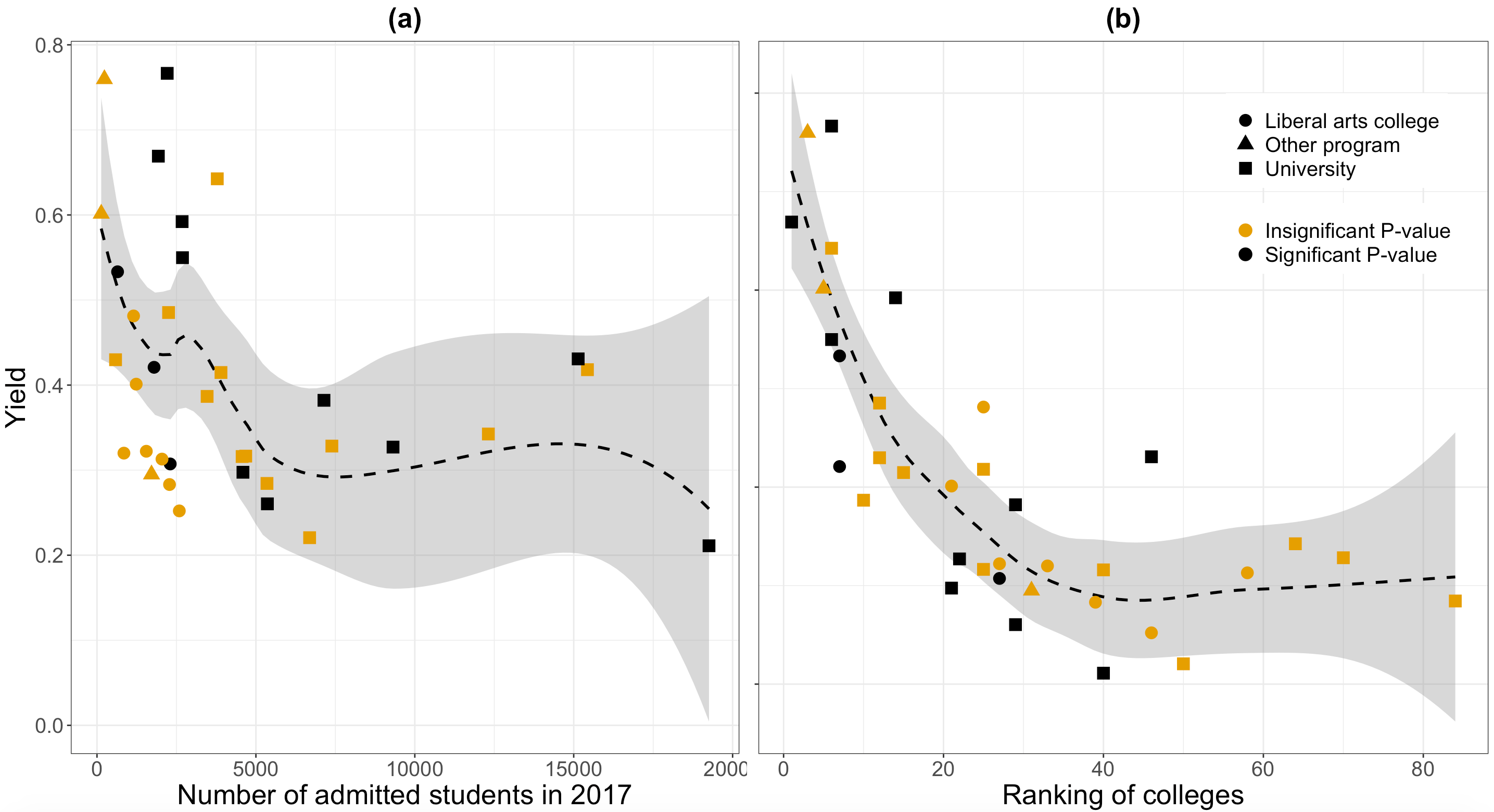}
    \caption{Regression of the yield on (a) the size of admitted class and (b) the ranking of colleges. 
    We fit the dashed curves and the $95\%$ confidence bands using smoothing splines.
     The ranking is according to \emph{U.S. News and World Report},
     where two (or more) colleges might tie, and liberal arts colleges, national universities, and other undergraduate programs  are ranked separately within their categories.}
    \label{fig:nytdata}
\end{figure}
\subsection{Simulated Graduate School Admissions}
\label{sec:hierarchpref}

\begin{figure}[!ht]
\centering
\includegraphics[width=0.9\textwidth, height=2.25in]{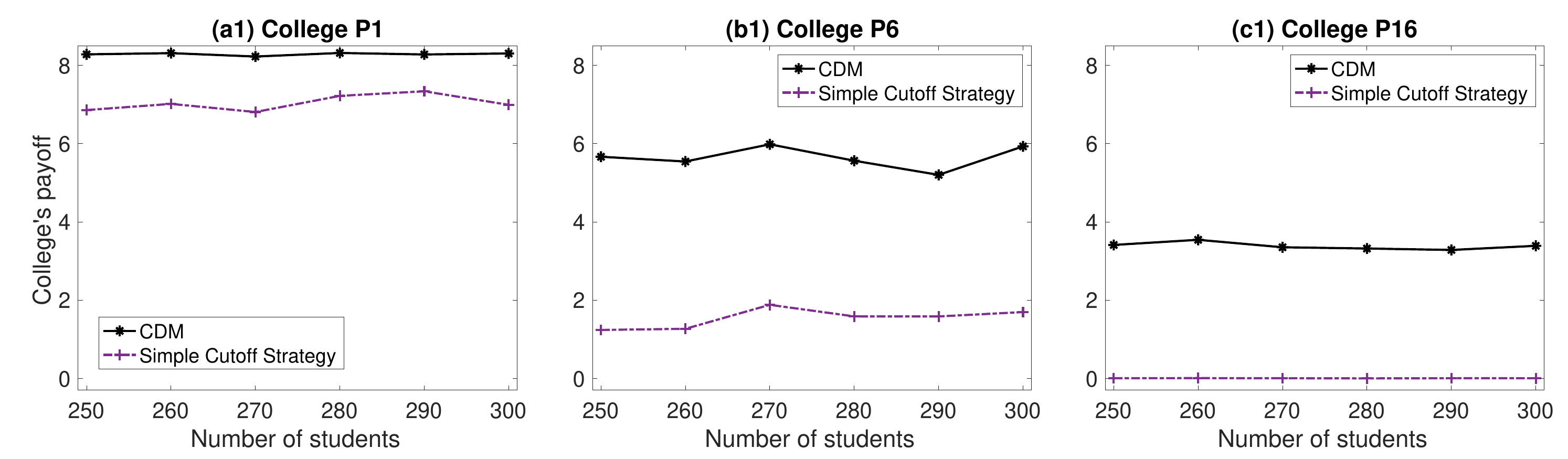}
\includegraphics[width=0.9\textwidth, height=2.25in]{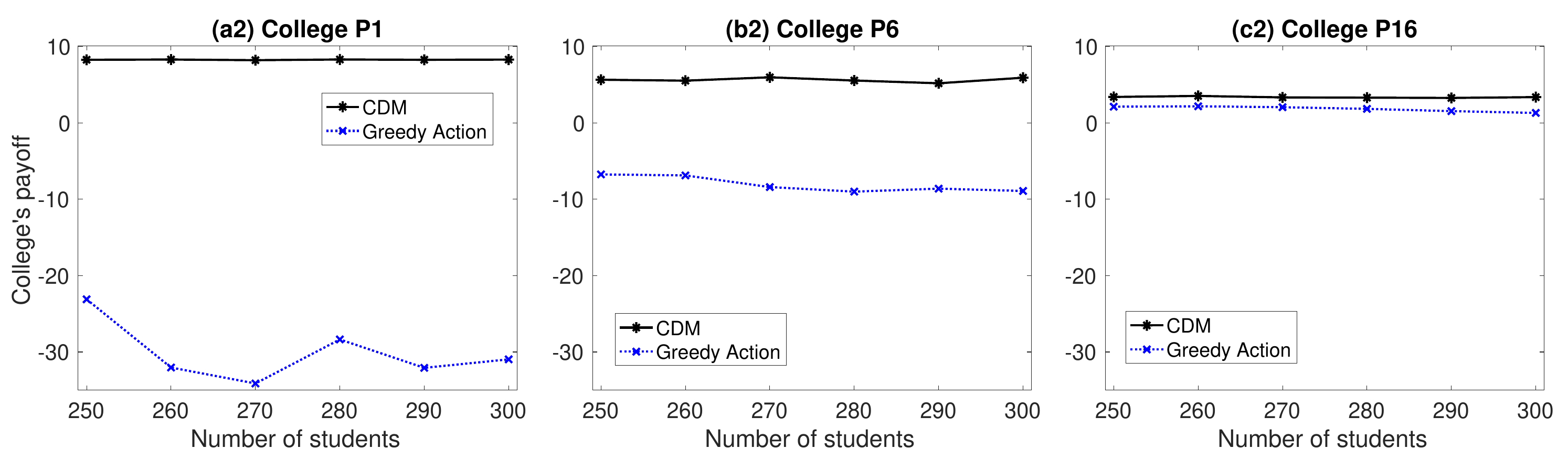}
\caption{Performance of three strategies with varying numbers of students. The results are averaged over $500$ data replications. (a1) and (a2): College $P_1$ from tier 1. (b1) and (b2): College $P_6$ from tier 2. (c1) and (c2): College $P_{16}$ from tier 3.}
\label{fig:eg4_1}
\end{figure}

We study a simulated graduate school admission problem, where each graduate has a limited quota. Suppose there are a total of $m=50$ graduate schools from three tiers of colleges: five top colleges $\{P_1,\ldots,P_5\}$, ten good colleges $\{P_6,\ldots,P_{15}\}$, and $35$ other colleges $\{P_{16},\ldots,P_{50}\}$. Each has the same quota $q=5$ and penalty $\gamma=2.5$. The simulation generates students' preferences with ten different states $\{s_1,\ldots,s_{10}\}\subset[0,1]$. For any state, students' preferences for colleges from the same tier are random. However,  students always prefer top colleges to good colleges, and then the other colleges.
The random preferences depend on the state due to colleges' uncertain reputation and popularity in the current year. 
We consider varying numbers of students  $n=\{250,260,270,280,290,300\}$. 
For each number of students, there are ten students having a score $v_j$ chosen uniformly and independently from $[0.9,1]$ and $100$ students having score $v_j$ drawn uniformly and independently from from $[0.7,0.9)$. The rest of the students have score $v_j$ chosen uniformly and independently from $[0,0.7)$. The fits $e_{ij}$ for all college-student pairs are drawn uniformly from $[0,1]$. The colleges' preferences are determined by students' latent utilities according to Eq.~(\ref{eqn:defofutility}). 
We consider a simple case that students face negligible application costs. 
Since students do not know how colleges evaluate their fits (e.g., personal essays), submitting applications to all colleges is students' dominant strategy.

We compare the colleges' expected payoffs  achieved by three methods: CDM, simple cutoff strategy, and greedy action. The latter two methods are described in Section \ref{sec:agentbetteroff}.  
We implement the CDM by Algorithm \ref{alg:decentralizedcdm}.
The training data are simulated from colleges’ random proposing, where colleges admit random numbers of students according to their latent utilities.
The training data consists of $20$ draws of random proposals under each of the students' preference structures. In total, the sample size $T=200$.
The testing data draws a random state from $\{s_1,\ldots,s_{10}\}$, which gives the corresponding student preferences.  
Figure \ref{fig:eg4_1} reports the averaged payoffs of three colleges $P_1,P_6$ and $P_{16}$ over $500$ data replications. Here colleges $P_1,P_6$ and $P_{16}$ belong to the three different tiers, respectively.
In Figure \ref{fig:eg4_1}(a1)---(a2), all  colleges other than $P_1$ use the CDM with mean calibration, while $P_1$ uses one of the three methods. The same setup applies to Figure \ref{fig:eg4_1}(b1)---(b2) and  Figure \ref{fig:eg4_1}(c1)---(c2). 
It is seen that the CDM gives the largest payoffs for $P_1,P_6$ and $P_{16}$. The CDM performs particularly well for tier 2 and tier 3 colleges  compared to the simple cutoff strategy. 
Moreover, the CDM outperforms the greedy action, especially for tier 1 and 2 colleges. The reason is that the greedy action does not calibrate the state and enrolls too many students.

\section{Discussion}
\label{sec:discussion}
We have presented a statistical framework for decision-making under uncertainty, in settings when preferences are unknown a priori and must be learned from data, and when the output decisions involve the allocation of shared resources. Our approach is decentralized, and thus applicable to agents who may not wish to communicate their preferences or share their decisions. 
In the model, arms have uncertain preferences that depend on the unknown state of the world and the arms' acceptance probabilities also depend on the competition among the agents. 
We construct a minimax optimal estimator for the learning of acceptance probabilities under this model. We propose a specific algorithm, the CDM algorithm, for learning optimal strategies from historical data. The CDM algorithm calibrates the unknown state by perturbing the state and balancing the marginal utility and the marginal penalty for exceeding the capacity. This calibration procedure takes opportunity costs into account. We  show that CDM makes it safe for agents to act straightforwardly on their preferences. The CDM procedure achieves stability under incomplete information, where the formulation of the individual rationality condition models agents' beliefs. 
Moreover, CDM is fair for arms in the sense that arms have no justified envy. 


It is possible to extend the CDM framework to other competition structures under a scarcity of shared resources. 
It is of interest to study multi-stage decentralized matching, which has applications on college admissions with a waiting list. 
However, multi-stage matching involves new challenges. For example, it requires a different model to incorporate the new hierarchical structure arising from the multi-stage decision-making, where arms available at subsequent stages are likely to be worse than those available at the current stage. Hence,  the optimal strategy and economic implications in multi-stage matching will be different from those in single-stage matching, and the analysis would be more technical and tedious. We are currently working on that extension.
Another possible extension is to 
consider algorithmic strategies where agents' preferences exhibit complementarities in decentralized matching. For instance, some firms demand workers that complement one another in terms of their skills and roles. 
Another interesting extension is to consider decentralized markets with indifferent preferences.
For example, many applicants may be indistinguishable for a college. Still, it is necessary to break ties since the college may have insufficient capacity to admit all applicants in the same indifference class.

\acks{The authors thank Robert M. Anderson, two anonymous reviewers, and the Editor for their invaluable feedback on an earlier version of this work.
}

\vskip 0.2in
\bibliography{match}

\newpage

\appendix


\section{Additional Example, Algorithm and Simulation}

\subsection{Three-Agent Model}

We give an additional example of the three-agent model and show the explicit form of the acceptance probability $\pi_i(s_i,v_j),i=1,2,3$ for this model.

\begin{example}
\label{eg:three-agent}
Consider a three-agent model with agents $P_1,P_2$ and $P_3$. Suppose that $P_1$ pulls an arm $A_j\in\mathcal A$.
Denote by $\mu_i(s_1)$  the probability that an arm prefers $P_i$ under the state $s_1$, $i=1,2,3$, where $\mu_1(s_1)+\mu_2(s_1)+\mu_3(s_1)=1$. 
Let $\sigma_i$ be $P_i$'s strategy, which is defined by $\sigma_i(v_j,e_{2j})=\mathbf 1(P_i \text{ pulls } A_j \text{ with attributes } \{v_j,e_{ij}\})$, $i=2,3$.
Note that  $A_j$ would be pulled  by  $P_1$ with probability $p_1\equiv [1-\sigma_{2}(v_j,e_{2j})][1-\sigma_{3}(v_j,e_{3j})]$, and  pulled by both $P_1$ and $P_2$ with probability $p_2\equiv \sigma_{2}(v_j, e_{2j})[1-\sigma_{3}(v_j,e_{3j})]$, and  pulled by both $P_1$ and $P_3$ with probability $p_3\equiv[1-\sigma_{2}(v_j,e_{2j})]\sigma_{3}(v_j, e_{3j})$,  and  pulled by all of $P_1,P_2$ and $P_3$ with probability $p_4\equiv\sigma_{2}(v_j, e_{2j})\sigma_{3}(v_j,e_{3j})$. Hence $A_j$ would accept  $P_1$ with probability 
\begin{equation*}
    p_1+\frac{\mu_1(s_1)}{\mu_1(s_1)+\mu_2(s_1)}p_2+\frac{\mu_1(s_1)}{\mu_1(s_1)+\mu_3(s_1)}p_3 + \mu_1(s_1)p_4.
\end{equation*}
Since $e_{ij}$ are independent and unknown to $A_j$, and $\sigma_i$ depends on $e_{ij}$, $i=2,3$, the expected probability that $A_j$ accepts $P_1$ is as follows,
\begin{equation*}
\begin{aligned}
    \pi_1(s_1,v_j) & = \E\left[p_1+\frac{\mu_1(s_1)}{\mu_1(s_1)+\mu_2(s_1)}p_2+\frac{\mu_1(s_1)}{\mu_1(s_1)+\mu_3(s_1)}p_3 + \mu_1(s_1)p_4\right],\\
    & = \left[1-\E_{e_{2j}}[\sigma_{2}(v_j,e_{2j})]\right]
    \left[1-\E_{e_{3j}}[\sigma_{3}(v_j,e_{3j})]\right]\\
    &\quad\quad + \frac{\mu_1(s_1)}{\mu_1(s_1)+\mu_2(s_1)}\E_{e_{2j}}[\sigma_{2}(v_j,e_{2j})]
    \left[1-\E_{e_{3j}}[\sigma_{3}(v_j,e_{3j})]\right]\\
    &\quad\quad + \frac{\mu_1(s_1)}{\mu_1(s_1)+\mu_3(s_1)}\left[1-\E_{e_{2j}}[\sigma_{2}(v_j,e_{2j})]\right]\E_{e_{3j}}[\sigma_{3}(v_j,e_{3j})]\\
    &\quad\quad + \mu_1(s_1)\E_{e_{2j}}[\sigma_{2}(v_j,e_{2j})]
    \E_{e_{3j}}[\sigma_{3}(v_j,e_{3j})].
\end{aligned}
\end{equation*} 
Moreover, we can derive an explicit form of $\pi_i(s_i,v_j),i=2,3$ by following a similar argument as above.
\end{example}
\noindent
Compared to Example \ref{eg:two-agent} in Section \ref{sec:stateofnature}, the acceptance probability $\pi_1(s_1,v_j)$ in the three-agent model is different from that in the two-agent model, unless $\mu_3(s_1) = 0 $. 
Hence the acceptance probability would  change  with the  number of agents competing for the arm.

\subsection{Algorithm of CDM under the Maximin Calibration}
\label{sec:CDMwithmaximin}

We provide the algorithm for the calibrated decentralized matching (CDM) procedure under the maximin calibration in Theorem \ref{cor:introofvalueV}. The CDM under the maximin calibration consists of five main steps. Step 1 is to estimate the acceptance probability based on the historical data, where we obtain the estimate $\widehat{\pi}_i(s_i,v) =[1+\exp(-\widehat{f}_i(s_i,v))]^{-1}$ by Eq.~(\ref{eqn:estpiisvj}). 
Step 2 is to construct the cutoff $\widehat{e}_i(s_i,v)$ for a given state $s_i$ according to Eq.~(\ref{eqn:choiceofeim}).
Step 3 is to calculate the arm set $\widehat{\mathcal B}_i(s_i)$ for the given  $s_i$ according to Eq.~\eqref{eqn:armsethatbisi}.
In addition, Steps 2 to 3 are carried out repeatedly over a number of states. 
Step 4 is to calibrate the state for maximizing the agent’s worst-case expected payoff, which calibration is referred to as the maximin calibration.
Given the estimator $\widehat{\pi}_i(s_i,v)$ and the arm set $\widehat{\mathcal B}_i(s_i)$, the calibration in Theorem \ref{cor:introofvalueV} can be written the solution to the following equation,
\begin{equation}
\label{eqn:implementmaximincalibration}
\begin{aligned}
& \sum_{j\in\widehat{\mathcal B}_i(s_i)}(v_j+e_{ij})\cdot\left[\widehat{\pi}_i(1,v_j)-\widehat{\pi}_i(0,v_j)\right] - \gamma_i\sum_{j\in\widehat{\mathcal B}_i(s_i)}\widehat{\pi}_i(1,v_j) + \gamma_iq_i
\\
&\quad\quad\quad= \sum_{j\in\widehat{\mathcal B}_i(1)}(v_j+e_{ij})\cdot\widehat{\pi}_i(1,v_j) - \sum_{j\in\widehat{\mathcal B}_i(0)}(v_j+e_{ij})\cdot\widehat{\pi}_i(0,v_j).
\end{aligned}
\end{equation}
Finally, Step 5 is to produce the final arm set under the calibrated state. 
We summarize the CDM under the maximin calibration in Algorithm \ref{alg:decentralizedcdmmaximin}.

\begin{algorithm}
\caption{ \normalsize{Calibrated decentralized matching (CDM) under the maximin calibration}}\label{alg:decentralizedcdmmaximin}
\begin{algorithmic}[1]
\State  \normalsize{\textbf{Input:} Historical data $\mathcal D=\{(s_i^t,v_j^t,e_{ij}^t,y_{ij}^t):i\in[m]; j\in \mathcal B_{i}^t; t\in[T]\}$; New arm set $\mathcal A^{T+1}$  with attributes $\{(v_j,e_{ij}):i\in[m];j\in[n]\}$ at time $T+1$; Penalty $\{\gamma_i:i\in[m]\}$ for exceeding the quota.}
\State \textbf{for} $i=1,2,\ldots, m$ \textbf{do}
\State \quad \textbf{Step 1:} Predict the acceptance $\widehat{\pi}_i(s_i,v_j)$ by Eq.~(\ref{eqn:estpiisvj});
\State \quad  \textbf{for} $s_i$ in a linearly spaced vector in $[0,1]$  \textbf{do}
\State \quad\quad \textbf{Step 2:} Construct the cutoff $\widehat{e}_i(s_i,v)$ by Eq.~(\ref{eqn:choiceofeim});
\State \quad\quad \textbf{Step 3:} Calculate the arm set $\widehat{\mathcal B}_i(s_i)$ by Eq.~\eqref{eqn:armsethatbisi};
\State \quad \textbf{end for}
\State \quad \textbf{Step 4:} Calibrate the state $s_i$ such that the difference between LHS and RHS of Eq.~\eqref{eqn:implementmaximincalibration} is below a pre-specified tolerance level;
\State \quad \textbf{Step 5:} Calculate the arm set $\widehat{\mathcal B}_i(s_i)$ by Eq.~\eqref{eqn:armsethatbisi} with the calibrated state $s_i$;
\State \textbf{end for}
\State \textbf{Output:} The arm sets $\widehat{\mathcal B}_1(s_1), \widehat{\mathcal B}_2(s_2),\ldots, \widehat{\mathcal B}_m(s_m)$ for agents.
\end{algorithmic}
\end{algorithm}


\subsection{Sensitivity Analysis}
\label{sec:sensitivity}

\begin{figure}[!ht]
\centering
\includegraphics[width=\textwidth, height=2.25in]{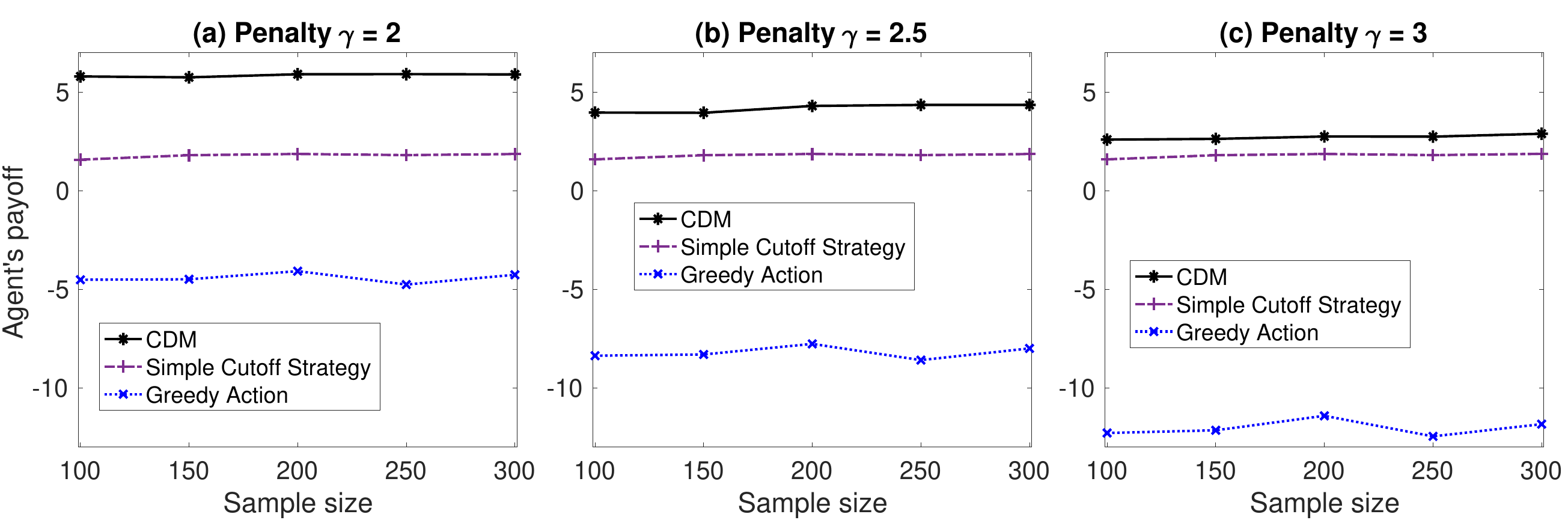}
\caption{Performance of three strategies with varying numbers of sample size. The results are averaged over $500$ data replications. Penalty levels (a) $\gamma=2$, (b) $\gamma=2.5$, (c) $\gamma=3$.}
\label{fig:sensitivity_1}
\end{figure}

\begin{figure}[!ht]
\centering
\includegraphics[width=\textwidth, height=2.25in]{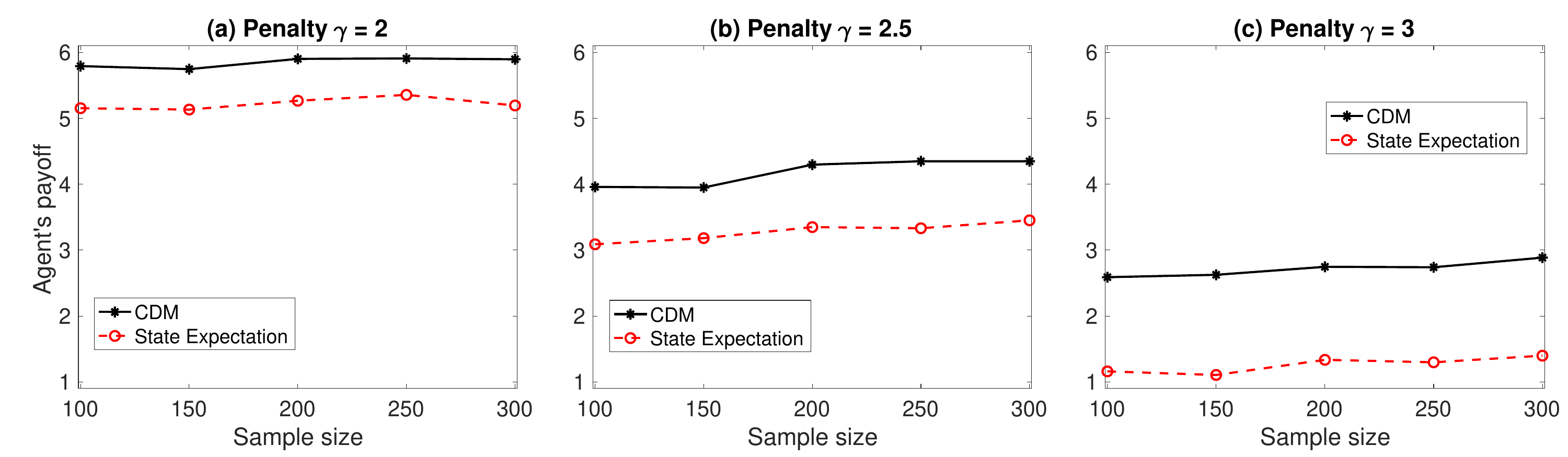}
\caption{Performance of two calibration methods with varying numbers of sample size. The results are averaged over $500$ data replications. Penalty levels (a) $\gamma=2$, (b) $\gamma=2.5$, (c) $\gamma=3$.}
\label{fig:sensitivity_2}
\end{figure}

We next carry out a sensitivity analysis of the CDM with respect to the sample size of training data. We adopt the same simulation setup as in Section \ref{sec:agentbetteroff} with $m=10$ agents and $n=150$ arms. 
Each arm has a score $v_j$ and fits $e_{ij}$ drawn uniformly from $[0,1]$.
The agents' preferences are determined by arms' latent utilities as in Eq.~(\ref{eqn:defofutility}).
Each agent has the same quota $q=5$ and the same penalty $\gamma$ chosen from $\{2,2.5,3\}$. 
The simulation generates random arm preferences with $10$ different states from $\{s_1,\ldots,s_{10}\}\subset [0,1]$. The training data are simulated by having agents pull random numbers of arms according to their latent utilities. 
We vary the sample size of training data $T=\{100,150,200,250,300\}$, where 
the training data consists of $T/10$ rounds of random proposals under each of the arms' preference structures. 
The testing data draws a random state from $\{s_1,\ldots,s_{10}\}$ and generates the corresponding arms' preferences. 
We compare the CDM with alternative methods: (i) the \emph{simple cutoff strategy} where the agent chooses $q$ best arms according to the latent utility; (ii) the \emph{greedy action} where the agent chooses arms with maximum expected  utilities and a total expected acceptance up to $q$; (iii) the \emph{state expectation} that calibrates the unknown state using the naive mean estimate of states.

Figures \ref{fig:sensitivity_1} and \ref{fig:sensitivity_2} report the agent $P_1$'s averaged payoffs with the varying sample sizes $T$ based on $500$ data replications. 
Here, all  agents other than $P_1$ use the CDM with mean calibration in Algorithm \ref{alg:decentralizedcdm}. On the other hand,  $P_1$ uses one of the four methods: CDM, simple cutoff strategy, greedy action, and state expectation.
It is seen that the CDM consistently outperforms alternative methods.
Moreover, the CDM performs relatively robust with respect to the sample size $T$.


\section{Technical Proofs}
\subsection{ANOVA Decomposition for General Utility Function}
\label{sec:anovautility}

Suppose that each arm $A_j\in\mathcal A$ is described by $x_j$ and  $(\epsilon_{1j}\trans,\epsilon_{2j}\trans,\ldots,\epsilon_{mj}\trans)\trans$, where $x_j$ is a multidimensional vector available to all agents and $\epsilon_{ij}$ is a multidimensional vector that is only available to agent $P_i$. 
In college admissions,  $x_j$ can be student $A_j$'s high school record and test score on a nationwide exam such as SAT/ACT, and $\epsilon_{ij}$ represents student $A_j$'s college-specific essays or test scores. 

\begin{proposition}
For any agent-specific latent utility function $U_i$, the ANOVA decomposition allows the separable form
\begin{equation*}
U_i(A_j) = v_j + e_{ij}, \quad\forall  i\in[m]\text{ and }j\in[n].
\end{equation*}
Here, $v_j\in\R$ is a function of $x_j$ and is common to all agents. The agent-specific $e_{ij}\in\R$  is a function of $x_j$ and $\epsilon_{ij}$ and is considered only by agent $P_i$. Thus, the separable structure of the utility function in Eq.~(\ref{eqn:defofutility}) can be assumed without loss of generality.
\end{proposition}
\begin{proof}
Denote the utility function $U_i(A_j) \equiv U_i(x_j,\epsilon_{ij})\in\R$.
By the analysis of variance (ANOVA)  decomposition,  we have that for $P_i\in\mathcal P, A_j\in\mathcal A$, 
\begin{equation}
\label{eqn:anovautility}
\begin{aligned}
& v_j \equiv \frac{1}{m}\sum_{i=1}^m\E_{\epsilon_{ij}}[U_i(x_j,\epsilon_{ij})]\quad \text{ and }\\
& e^\dagger_{ij} \equiv \E_{\epsilon_{ij}}[U_i(x_j,\epsilon_{ij})] -  \frac{1}{m}\sum_{i=1}^m\E_{\epsilon_{ij}}[U_i(x_j,\epsilon_{ij})], \quad e^\ddagger_{ij} \equiv U_i(x_j, \epsilon_{ij}) - \E_{\epsilon_{ij}}[U_i(x_j, \epsilon_{ij})].
\end{aligned}
\end{equation}
Here, $v_j$ represents the average utility of  $x_j$ and is common to all agents.
The $e_{ij}^{\dagger}$ is agent $P_i$'s adjustment for the utility of  $x_j$. The $e_{ij}^{\ddagger}$ is the
utility of  $\epsilon_{ij}$ received by agent $P_i$. Thus,  $e_{ij}^{\dagger}$ and $e_{ij}^{\ddagger}$ are agent-specific  and they are only known to agent $P_i$.
Moreover, letting  $e_{ij}\equiv e_{ij}^{\dagger} + e_{ij}^{\ddagger}$, then Eq.~(\ref{eqn:anovautility}) implies the desired result that  $U_i(x_j,\epsilon_{ij}) = v_j + e_{ij}$.
We refer to Figure \ref{fig:scorefitutility} for an illustration on the ANOVA decomposition.
\begin{figure}[!ht]
    \centering
    \includegraphics[width=0.7\textwidth]{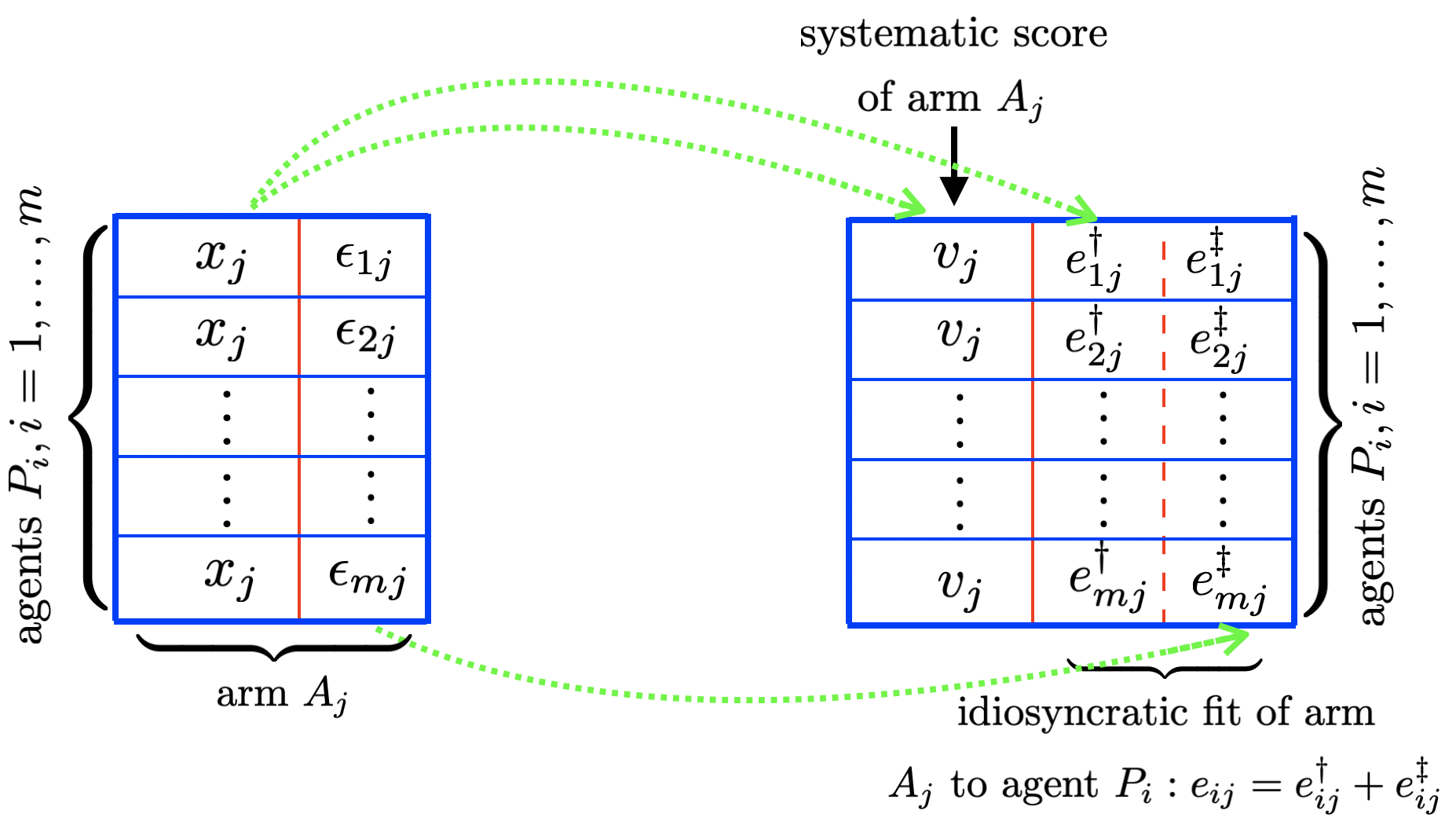}
    \caption{ANOVA decomposition of the utility function corresponding to Eq.~(\ref{eqn:anovautility}).}
    \label{fig:scorefitutility}
\end{figure}

The utility function $U_i(x_j,\epsilon_{ij})$ is generally assumed to be strictly increasing in $x_j$ and nondecreasing in $\epsilon_{ij}$ 
\citep{che2016,lee2016}.
Then $v_j$ in Eq.~(\ref{eqn:anovautility}) is strictly increasing in $x_j$, and $e_{ij}$ is nondecreasing in $\epsilon_{ij}$, for any $i\in[m]$.
\end{proof}

We demonstrate the ANOVA decomposition  in Eq.~(\ref{eqn:anovautility})  through the college admission example. The score
 $v_j$ in  Eq.~(\ref{eqn:anovautility})  represents the ``public valuation" of a high school's quality, high school GPA, and SAT/ACT score. 
 Although most colleges place  the highest importance on academic achievement in evaluating applications, each factor's specific weight can differ from college to college. Moreover, ``holistic admission"---such that a high SAT/ACT score and a high GPA is no guarantee of admission---is not rare in college admissions, especially for top colleges. 
The decomposition in Eq.~(\ref{eqn:anovautility}) incorporates the college-specific weight for students' academic performance and extracurricular activity to terms $e_{ij}^\dagger$ and $e_{ij}^{\ddagger}$. 
The term $e_{ij}^\dagger$ in  Eq.~(\ref{eqn:anovautility}) can represent college-specific adjustment, which is  ``private valuation" of a student's high school record and SAT/ACT score, in addition to special talent, grades of challenging college preparatory curriculum and work experience. For example, some colleges may place a larger weight on the SAT/ACT score  than other colleges.
The term $e_{ij}^{\ddagger}$  in Eq.~(\ref{eqn:anovautility}) is a college-specific ``private valuation" of students' writing skills and compelling personal stories.

\subsection{Proof of Theorem \ref{prop:accepwelldefine}: Acceptance Probability}
\begin{proof}
The uncertainty of the arm's acceptance mainly consists of two parts: the state of the world and the agents' strategies. 
We employ a four-step construction approach to show that there exists a probability mass function (PMF) $\pi_i(s_i,v_j)$ characterizing the probability of $A_j$ accepting $P_i$.
First, Nature draws the state $\omega$ and arms' preferences for agents that are characterized by $\omega$. Then agent-specific states are determined as a function of $\omega$: $s_i = s_i(\omega)$ for $i\in[m]$. 

Second, we derive an arm's acceptance probability.
Suppose now agent $P_i$ pulls  arm $A_j$. Since $A_j$ would accept its most preferred agent among those who have pulled it, $A_j$'s acceptance of $P_i$ depends on other agents' strategies.
Let $I\subseteq [n]$ and  $P_{I} \equiv \{P_k, k\in I\}$, which is the set of agents that pull arm $A_j$.
We define
\begin{equation*}
\mu_{k,I\cup\{i\}}(\omega, v_j,\mathbf e_j)\equiv\mathbb P(A_j\text{ accepts }P_k\ | \ P_{I\cup\{i\}} \text{ pulls } A_j),\quad k\in I\cup\{i\}.
\end{equation*} 
That is, $\mu_{k,I\cup\{i\}}(\omega, v_j,\mathbf e_j)$ is the PMF that an arm with the same score and fits as $A_j$ would accept $P_k$ conditional on the event that agents $P_{I\cup\{i\}}$ have pulled the arm. The  $\mu_{i,I\cup\{i\}}(\omega, v_j,\mathbf e_j)$ depends on arms' preferences as characterized by $\omega$, and also on agents' strategies determined by the score $v_j$ and the vector of fits $\mathbf e_j = (e_{1j},\ldots,e_{mj})$.
Moreover, the PMF $\mu_{k,I\cup\{i\}}(\omega, v_j,\mathbf e_j)$ satisfies that,
\begin{equation}
\label{eqn:jointprobmu}
\sum_{k\in I\cup\{i\}}\mu_{k,I\cup\{i\}}(\omega,v_j, \mathbf e_j) = 1,\quad \forall j\in[n].
\end{equation}

Third, we derive an arm's acceptance probability from an agent $P_i$'s perspective. Suppose again that agent $P_i$ pulls  arm $A_j$. Define the PMF $\pi_i(s_i,v_j)$ as,
\begin{equation*}
\pi_i(s_i,v_j) \equiv \mathbb P(A_j\text{ accepts }P_i\ |\ P_i \text{ pulls }A_j).
\end{equation*}
Let $\mathcal I_{-i}\equiv \{I: I\subseteq [n]\setminus\{i\}\}$ be the family of subsets of $[n]\setminus\{i\}$. Then by Tonelli's theorem, we can write
\begin{equation}
\label{eqn:probajacceptspi}
\begin{aligned}
& \pi_i(s_i,v_j)\\
& =\mathbb E\left[\mathbb P(A_j\text{ accepts }P_i\ |\ P_{I\cup\{i\}} \text{ pulls }A_j, P_{[n]\setminus (I\cup\{i\})} \text{ do not pull }A_j)\right]\\
& = \mathbb E_{I\in\mathcal I_{-i}}\left[\mathbb E_{\mathbf e_{-i,j}}[\mu_{i,I\cup\{i\}}(\omega,v_j,\mathbf e_j)\cdot \mathbf 1(P_{I}\text{ pulls }A_j)\cdot \mathbf 1(P_{[n]\setminus(I\cup\{i\})}\text{ do not pull } A_j)]\right],
\end{aligned}
\end{equation}
where $\mathbf e_{-i,j}=(e_{1j},\ldots,e_{i-1,j},e_{i+1,j},\ldots,e_{n,j})$. When conditional on $v_j$, $\mathbf 1(P_I\text{ pulls }A_j)$ only depends on $\{e_k,k\in I\}$, and similarly, $\mathbf 1(P_{[n]\setminus(I\cup\{i\})}\text{ do not pull }A_j)$ only depends on $\{e_k,k\in [n]\setminus(I\cup\{i\})\}$.
Then the defined PMF $\pi_i(s_i,v_j) $ represents the PMF $\mu_{i,I\cup\{i\}}(\omega,v_j,\mathbf e_j)$ being averaged over agents $P_I$' strategies. Since $\pi_i$ is defined by conditioning on $P_i$'s strategy (i.e., $P_i$ pulls $A_j$), the probability $\pi_i(s_i,v_j)$ does not depend on the fit $e_{ij}$. Moreover, if knowing that agents $P_{I\cup\{i\}}$ pull arm $A_j$, then by Eqs.~(\ref{eqn:jointprobmu}) and (\ref{eqn:probajacceptspi}), the PMF $\pi_i(s_i,v_j)$ satisfies that,
\begin{equation*}
\sum_{k\in I\cup\{i\}}\pi_k(s_k,v_j) =1,\quad \forall j\in[n].    
\end{equation*}

Fourth, we note that agent-specific states $s_i = s_i(\omega)$ is chosen by agent $P_i$ for $i\in[m]$. Hence without loss of generality, let $s_i$ is parameterized such that the PMF $\pi_i(s_i,v_j)$ is increasing in $s_i$.

Finally, the expected utility  that agent $P_i$ receives from pulling arm $A_j\in\mathcal A$ is
\begin{equation*}
\begin{aligned}
\E[\text{utility}] & = \E[\text{utility} \ | \ \text{successful pulling}]\cdot \mathbb P(\text{successful pulling}) \\
& = (v_j+e_{ij}) \cdot \pi_{i}(s_i, v_j).
\end{aligned}
\end{equation*}
This completes the proof.
\end{proof}

\begin{remark}
We now compare the functions $\mu_{i,I\cup\{i\}}(\omega, v_j,\mathbf e_j)$ and $\pi_i(s_i,v_j)$ from a learning perspective. 
It is shown in the above proof that  the function $\mu_{i,I\cup\{i\}}(\omega, v_j,\mathbf e_j)$ captures the distribution of arms' preferences. However, $\mu_{i,I\cup\{i\}}(\omega, v_j,\mathbf e_j)$ is imperfect for learning using data due to two reasons. First, if the agent $P_i$ wants to estimate $\mu_{i,I\cup\{i\}}(\omega, v_j,\mathbf e_j)$, $P_i$ has to identify other agents who are also pulling $A_j$, that is, to identify the set $I$.
However, thanks to the nature of decentralized markets, there is little communication among agents, and $P_ I$ cannot learn about other agents' strategies. Second, note that $\mu_{i,I\cup\{i\}}(\omega, v_j,\mathbf e_j)$ depends on the fit vector $\mathbf e_{-i,j}=(e_{1j},\ldots,e_{i-1,j},e_{i+1,j},\ldots,e_{n,j})$.
However, since the fits are private information, the vector $\mathbf e_{-i,j}$ is also unknown  to agent $P_i$.

In contrast, the learning of function $\pi_i(s_i, v_j)$ does not require $P_i$ knowing other agents' strategies. Moreover,  $\pi_i(s_i, v_j)$ does not depend on the  fit vector $\mathbf e_{-i,j}$. Therefore, the function $\pi_i(s_i, v_j) $ is favorable than the function $\mu_{i,I\cup\{i\}}(\omega, v_j,\mathbf e_j)$ in the sense that agent $P_i$ can use historical data to learn $\pi_i(s_i, v_j)$.
\end{remark}

\subsection{Proof of Theorem  \ref{thm:optestpi}: Optimal Estimation of Acceptance}
\label{appen:pfthmoae}
\begin{proof}
Let $\widetilde{f}_i$ be the minimizer of Eq.~(\ref{eqn:minklr}), that is,
\begin{equation*}
\widetilde{f}_i =\underset{f_i\in \mathcal H_{K_i}}{\arg\min}\left\{\sum_{t=1}^T\sum_{j\in \mathcal B_{i}^t}\left[-y_{ij}^tf_i(s_i^t,v^t_j)+\log\left(1+\exp\left(f_i(s_i^t,v^t_j)\right)\right)\right]+\frac{1}{2}\sum_{t=1}^Tn_{it}\lambda_i \|f_i\|_{\mathcal H_{K_i}}^2\right\}.
\end{equation*}
By  the results in  Chapters 5  of \citet{lin1998}, we obtain that
\begin{equation*}
\E_{s_i,v}[(\widetilde{f}_i - f_i)^2]\leq c_1 \left[T(\log T)^{-1}\right]^{-2r/(2r+1)}\quad \text{ as } T\to\infty,
\end{equation*}
where $\lambda_i\leq c_{\lambda}[T(\log T)^{-1}]^{-2r/(2r+1)}$. Here, $ c_{\lambda}, c_1>0$ are constants independent of $T$. Moreover, the estimate $\widetilde{f}_i$ is minimax rate-optimal. 
By the generalization properties of random features  \citep{rudi2017}, it is known that if the number of random features satisfies
\begin{equation*}
p\geq c_p[T(\log T)^{-1}]^{-2r/(2r+1)},
\end{equation*} 
then there exists some constant $c_2>0$ such that
\begin{equation*}
\E_{s_i,v}[(\widehat{f}_i - \widetilde{f}_i)^2]\leq c_2 \left[T(\log T)^{-1}\right]^{-2r/(2r+1)} \quad\text{ as } T\to\infty.
\end{equation*}
By the triangle inequality, there exists some constant $c_f>0$ such that
\begin{equation*}
\begin{aligned}
\E_{s_i,v}[(\widehat{f}_i - f_i)^2]& \leq  \E_{s_i,v}[(\widehat{f}_i - \widetilde{f}_i)^2] + \E_{s_i,v}[(\widetilde{f}_i -f_i)^2] \\
& \leq c_f\left[T(\log T)^{-1}\right]^{-2r/(2r+1)} \quad\text{ as } T\to\infty.
\end{aligned}
\end{equation*}
 Therefore, the estimate $\widehat{f}_i $  in Eq.~(\ref{eqn:predfisve}) is minimax rate-optimal. 
\end{proof}

\subsection{Proof of Theorem \ref{thm:cutoffaystate}: Cutoff Strategy} 
\label{subsec:pfcutstr}
\begin{proof}
We divide the proof of this theorem into five steps.
\paragraph{Step 1.} We show that the cutoff strategy with respect to the fits is optimal. Suppose that arms $A_{j_1}, A_{j_2}\in\mathcal A^{T+1}$ have the same score $v_{j_1}=v_{j_2}$, but $A_{j_1}$ has a worse fit than $A_{j_2}$ for agent $P_i$. Now assume that $A_{j_1}$ was pulled by $P_i$ but $A_{j_2}$ was not, that is, $A_{j_1}\in \widehat{\mathcal B}_i(s_i), A_{j_2}\not\in \widehat{\mathcal B}_i(s_i)$. 
Then the expected number of arms accepting $P_i$ is unchanged if $P_i$ replaces $A_{j_1}$ with $A_{j_2}$ in $\widehat{\mathcal B}_i(s_i)$. On the other hand, 
since the $P_i$'s expected payoff in Eq.~(\ref{eqn:totalutility}) is increasing in fit $e_{ij}$, $P_i$  has a strictly larger expected payoff if $P_i$ replaces $A_{j_1}$ with $A_{j_2}$. Hence, $P_i$ should pull $A_{j_2}$ instead $A_{j_1}$. This argument holds  regardless of strategies  of other agents. 

\paragraph{Step 2.} We prove that the cutoff $\widehat{e}_i(s_i,v)$  in Eq.~(\ref{eqn:choiceofeim}) is well-defined.   If the boundary $\{\mathcal B^{+}_i(s_i)\setminus\mathcal B^{-}_i(s_i)\}$ is not empty, the expected penalty due to exceeding the quota is
\begin{equation*}
\gamma_i\sum_{j\in\mathcal B^{+}_i(s_i)}\pi_i(s_i,v_j)-\gamma_iq_i.
\end{equation*}
The expected utility of pulling arms on the boundary is
\begin{equation*}
\sum_{j\in \mathcal B^{+}_i(s_i)\setminus\mathcal B^{-}_i(s_i)}(v_j+e_{ij})\cdot\pi_i(s_i,v_j).
\end{equation*}
Agent $P_i$ pulls an arm if the expected utility is at least the expected penalty, which justifies the condition specified by Eq.~(\ref{eqn:jbim1+}). 
Moreover, since $\widehat{e}_i(s_i,v)\in[0,1]$, we conclude  that the cutoff is well-defined.

\paragraph{Step 3.} We show that the cutoff strategy of pulling arms from the set $\widehat{\mathcal B}_{i}(s_{i})$ is near-optimal. 
Let $\widetilde{\mathcal B}_{i}(s_{i})\subseteq \mathcal A^{T+1}$ be any other arm set. Define the following mixed strategy:
\begin{equation*}
\sigma_{i,k}(s_{i,k},v, e_i;t) \equiv t\cdot\mathbf{1}\{(v,e_i) \in \widetilde{\mathcal B}_{i}(s_{i})\}+(1-t)\cdot\mathbf{1}\{(v,e_i) \in \widehat{\mathcal B}_{i}(s_{i})\},\quad\text{for } t\in[0,1].
\end{equation*}
The expected payoff of using the mixed strategy $\sigma_i$ is
\begin{equation*}
\begin{aligned}
\mathbb U_i(t)  \equiv & \sum_{j\in\mathcal A^{T+1}}(v_j+e_{ij})\cdot \pi_i(s_i,v_j)\cdot\sigma_i(s_i,v_j, e_{ij};t) \\
&\quad\quad\quad\quad\quad - \gamma_i\cdot\max\left\{\sum_{j\in\mathcal A^{T+1}}\pi_i(s_i,v_j)\cdot\sigma_i(s_i,v_j,e_{ij};t)-q_i,0\right\}.
\end{aligned}
\end{equation*}
It is clear that $\mathbb U_i(t)$ is concave in $t$.
We discuss the local change $d \mathbb U_i(0)/dt$ in three cases.

Case (I). Consider removing a single arm from $\widehat{\mathcal B}_i(s_i)$. If the arm is from the non-empty boundary $\{\mathcal B^{+}_i(s_i)\setminus\mathcal B^{-}_i(s_i)\}$,  the condition specified in Eq.~(\ref{eqn:jbim1+}) implies that $P_i$'s expected payoff will decrease if not pulling the arm. Moreover, since removing any other arms from $\widehat{\mathcal B}_i(s_i)$ results in a greater loss of $P_i$'s expected payoff compared to removing the arms on the non-empty boundary, we have $d \mathbb U_i(0)/dt< 0$ in this case.
By the concavity of $\mathbb U_i(t)$ in $t$,  we obtain
\begin{equation*}
\mathbb U_i(1)= \mathbb U_i(0)+\frac{d\mathbb U_i(0)}{dt}(1-0)< \mathbb U_i(0),
\end{equation*}

Case (II). Consider  adding a new arm to $\widehat{\mathcal B}_i(s_i)$, where the new arm has attributes $\{v_{j'},e_{ij'}\}$ and it is not from the set $\mathcal B_i^{+}(s_i)$. Denote by $\mathcal B_i'(s_i)$ the new arm set with the added arm. Note that $P_i$ pulls a new arm only if the arm has a larger expected utility than the expected penalty due to exceeding the quota, that is,
\begin{equation}
\label{eqn:vj'eij'g}
(v_{j'}+e_{ij'})\cdot\pi(s_i,v_{j'})\geq \gamma_i\sum_{j\in\mathcal B'_i(s_i)}\pi_i(s_i,v_j)-\gamma_iq_i.
\end{equation}
Since the new arm is not in $\mathcal B_i^{+}(s_i)$,  we have
\begin{equation}
\label{eqn:sumjipisvjqii}
\begin{aligned}
& \sum_{j\in\mathcal B'_i(s_i)}\pi_i(s_i,v_j)-q_i \\
& = \sum_{j\in\mathcal B'_i(s_i)}\pi_i(s_i,v_j)-\sum_{j\in\mathcal B_i^{+}(s_i)}\pi_i(s_i,v_j)+\sum_{j\in\mathcal B_i^{+}(s_i)}\pi_i(s_i,v_j)-q_i\\
 &\geq \sum_{j\in\mathcal B'_i(s_i)}\pi_i(s_i,v_j)-\sum_{j\in\mathcal B_i^{+}(s_i)}\pi_i(s_i,v_j)\\
 &\geq \pi_i(s_i,v_{j'}).
\end{aligned}
\end{equation}
Because that $\gamma_i>\sup_{j\in\mathcal A^{T+1}}\{v_j+e_{ij}\}$,  Eq.~(\ref{eqn:sumjipisvjqii}) is contradictory to  Eq.~(\ref{eqn:vj'eij'g}). 
Hence, adding a new arm to $\widehat{\mathcal B}_i(s_i)$ induces a loss in $P_i$'s expected payoff. Hence, $d \mathbb U_i(0)/dt< 0$ in this case. By the concavity of $\mathbb U_i(t)$ in $t$,  we obtain
\begin{equation*}
\mathbb U_i(1)= \mathbb U_i(0)+\frac{d\mathbb U_i(0)}{dt}(1-0)< \mathbb U_i(0),
\end{equation*}

Case (III): Consider  removing an arm with attributes $(v_j,e_{ij})$ from  $\widehat{\mathcal B}_{i}(s_{i})$ and simultaneously adding  new arms to  $\widehat{\mathcal B}_{i}(s_{i})$. Suppose that the new arms are from $\mathcal B_{i}''(s_{i})$ with attributes $(v_{j''},e_{ij''})$.  If $\widehat{\mathcal B}_{i}(s_{i}) = \mathcal B_{i}^-(s_{i})$, then the new arms are not in $\mathcal B_{i}^{-}(s_{i})$. By definition of $\mathcal B_{i}^{-}(s_{i})$, we have
$v_{j''} + e_{ij''}\leq \min_{j\in \mathcal B^-_{i}(s_{i})}(v_{j} + e_{ij}).$
Hence, 
\begin{equation}
\label{eqn:bdonb-}
\begin{aligned}
& \mathbb U_i(1)  -\mathcal U_i[\mathcal B^-_{i}(s_{i})] \leq 
 \sum_{j''\in\mathcal B_{i}''(s_{i})}(v_{j''}+e_{ij''}) \pi_{i}(s_{i},v_{j''})\\
&\leq \left[\min_{j\in \mathcal B^-_{i}(s_{i})}(v_j+e_{ij})\right]\cdot \left[q_i-\sum_{j\in \mathcal B^-_{i,k}(s_{i,k})}\pi_{i,k}(s_{i,k},v_j)\right]\\
& = \text{UE}^\dagger.
\end{aligned}
\end{equation}
If $\widehat{\mathcal B}_{i,k}(s_{i,k}) = \mathcal B_{i,k}^+(s_{i,k})$, by definition of $\mathcal B_{i,k}^+(s_{i,k})$
\begin{equation*}
\begin{aligned}
&\mathbb U_i(1)-\mathcal U_i[\mathcal B^+_{i}(s_{i})] \leq \mathbb U_i(1)-\mathcal U_i[\mathcal B^-_{i}(s_{i})]  \leq  \text{UE}^\dagger.
\end{aligned}
\end{equation*}
where the last inequality is by Eq.~(\ref{eqn:bdonb-}).
Hence,
\begin{equation*}
    \mathbb U_i(1) - \mathbb U_i(0) \leq  \text{UE}^\dagger.
\end{equation*}
Therefore, exchanging an arm in  $\widehat{\mathcal B}_{i,k}(s_{i,k})$  with arms not in  $\widehat{\mathcal B}_{i,k}(s_{i,k})$ could result in a decrease in the expected payoff $\mathbb U_i(t)$ by at most $\text{UE}^\dagger$. 
Combining the cases (I), (II), (III), we obtain that 
\begin{equation*}
\max_{\mathcal B_i\subseteq \mathcal A^{T+1}}\mathcal U_i[\mathcal B_i] - \mathcal U_i[\widehat{\mathcal B}_i(s_i)]\leq \text{UE}^\dagger.
\end{equation*}

\paragraph{Step 4.} We prove the other direction of the inequality. Since $\widehat{\mathcal B}_{i}(s_{i})\subseteq \mathcal A^{T+1}$,
\begin{equation*}
\max_{\mathcal B_i\subseteq \mathcal A^{T+1}}\mathcal U_i[\mathcal B_i] - \mathcal U_i[\widehat{\mathcal B}_i(s_i)]\geq 0.
\end{equation*}
\paragraph{Step 5.} If there is a continuum of arms and $\pi_{i}(\cdot,v)$ is continuous in $v$, then there exists $b_{i}\geq 0$
such that
$\Pi_{i,k}(b_{i,k}) = q_i$, where $\Pi_{i,k}(b_{i,k})$ is defined in Section \ref{sec:agentsexpectedpayoff}.
By definition of $\widehat{\mathcal B}_{i}(s_{i})$ we have, $\widehat{\mathcal B}_{i}(s_{i})=\mathcal B^+_{i}(s_{i}) = \mathcal B^-_{i}(s_{i})$,  and $\sum_{j\in \mathcal B^-_{i}(s_{i})}\pi_{i}(s_{i},v_j)=q_i$.
Hence $\text{UE}^\dagger = 0$. This completes the proof.
\end{proof}

\subsection{Proof of Theorem \ref{thm:optimalstates}: Mean Calibration for CDM}
\label{subsec:pfoffindstaave}
\begin{figure}[!ht]
    \centering
    \includegraphics[width=0.6\textwidth]{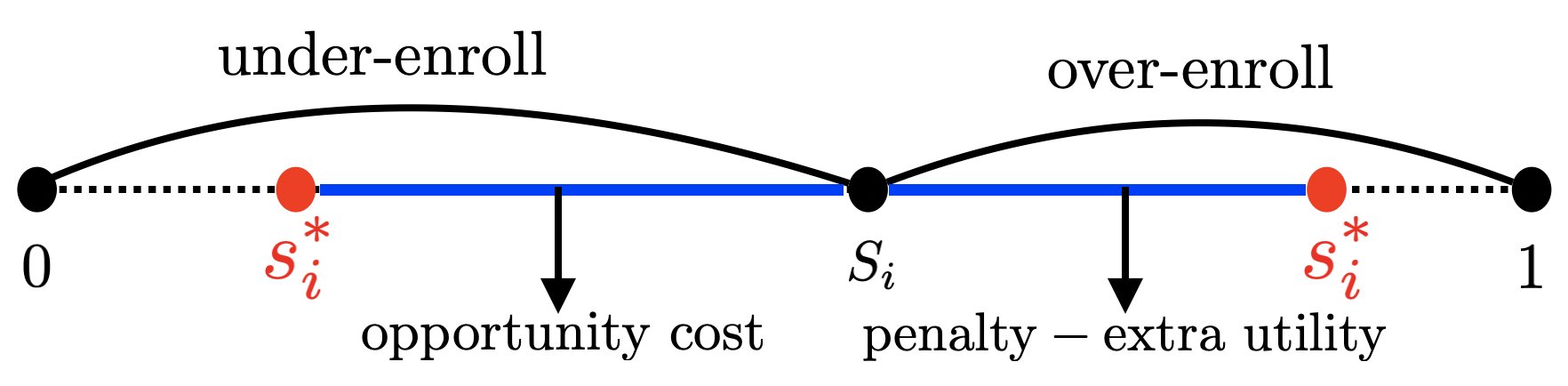}
    \caption{Cost of a strategy in the face of uncertain true state $s^*_i$.}
    \label{fig:cost1}
\end{figure}
\begin{proof}
By the proof of  Theorem \ref{thm:cutoffaystate},  $\widehat{\mathcal B}_i(s_i)\subseteq \widehat{\mathcal B}_i(s_i-\delta s_i)$ for any $\delta s_i\in(0,s_i)$. Hence, the marginal set $\partial \widehat{\mathcal B}_i(s_i)$  is well-defined. 
Let $V_i(s_i^*, \widehat{\mathcal B}_i(s_i))$ be the expected utility that $P_i$ receives by pulling arms from  $\widehat{\mathcal B}_i(s_i)$ and under the true state $s_i^*$. That is, 
\begin{equation*}
V_i(s_i^*,\widehat{\mathcal B}_i(s_i)) \equiv \sum_{j\in\widehat{\mathcal B}_i(s_i)}(v_j+e_{ij})\cdot\pi_i(s_i^*,v_j).
\end{equation*} 
Let $\mathcal N_i(s_i^*,\widehat{\mathcal B}_i(s_i))$ be the expected number of arms in $\widehat{\mathcal B}_i(s_i)$ accepting $P_i$ under $s_i^*$. That is,
\begin{equation*}
\mathcal N_i(s_i^*,\widehat{\mathcal B}_i(s_i)) \equiv \sum_{j\in\widehat{\mathcal B}_i(s_i)}\pi_i(s_i^*,v_j).
\end{equation*}
Similarly, we define $V_i(s_i^*,\partial\widehat{\mathcal B}_i(s_i)) $ and $\mathcal N_i(s_i^*,\partial\widehat{\mathcal B}_i(s_i))$ for the marginal  set $\partial\widehat{\mathcal B}_i(s_i)$.  
Let the marginal utility  be
\begin{equation*}
u_i(s_i^*, \partial\widehat{\mathcal B}_i(s_i)) \equiv\frac{V_i(s_i^*,\partial\widehat{\mathcal B}_i(s_i))}{\mathcal N_i(s_i^*,\partial\widehat{\mathcal B}_i(s_i))}.
\end{equation*}
Since the true state $s_i^*$ is unknown, the cost of pulling the arm set $\widehat{\mathcal B}_i(s_i)$
consists of two parts. See an illustration in Figure \ref{fig:cost1}.   The first  is the \emph{over-enrollment} cost (OE), which occurs if the calibration parameter $s_i<s_i^*$. Then the realized number of arms in  $\widehat{\mathcal B}_i(s_i)$ accepting $P_i$ will be greater than the expected number of arms in  $\widehat{\mathcal B}_i(s_i)$ accepting $P_i$. That is, $\mathcal N_i(s_i^*, \widehat{\mathcal B}_i(s_i))>\mathcal N_i(s_i,\widehat{\mathcal B}_i(s_i))=\mathcal N_i(s^*,\widehat{\mathcal B}_i(s_i^*))=q_i$. Thus, OE depends on $s_i$ and can be written as
\begin{equation*}
\begin{aligned}
&\quad \text{OE}(s_i) \\
&  \equiv \E_{s_i^*}\left[\gamma_i\left\{\left.\mathcal N_i(s_i^*,\widehat{\mathcal B}_i(s_i)) - \mathcal N_i(s_i^*,\widehat{\mathcal B}_i(s_i^*))\right\} - \left\{V_i(s_i^*,\widehat{\mathcal B}_i(s_i)) - V_i(s_i^*,\widehat{\mathcal B}_i(s_i^*))\right\}\ \right| \ s_i< s_i^*\leq 1\right]\\
& = \E_{s^*}\left[\left.\int_{t=s_i}^{s_i^*}[\gamma_i-u_i(s_i^*,\partial\widehat{\mathcal B}_i(t))]\cdot\mathcal N_i(s_i^*,\partial\widehat{\mathcal B}_i(t))
dt\ \right| \ s_i< s_i^*\leq 1\right].
\end{aligned}
\end{equation*}
Here, OE$(s_i)$ equals the penalty of the arms in  $\widehat{\mathcal B}_i(s_i)$ which would accept $P_i$ exceeding the quota, and deducts the utility of these arms.

The second part of the cost is \emph{under-enrollment} (UE), which occurs if the calibration parameter  $s_i>s_i^*$. Then $\mathcal N_i(s^*,\widehat{\mathcal B}_i(s_i))<\mathcal N_i(s_i^*,\widehat{\mathcal B}_i(s_i^*))=q_i$.
The UE$(s_i)$ equals the opportunity cost in the sennse that $P_i$ could have successfully pulled more arms:
\begin{equation*}
\begin{aligned}
\text{UE}(s_i)& \equiv \E_{s_i^*}[V_i(s_i^*,\widehat{\mathcal B}_i(s_i^*)) - V_i(s_i^*,\widehat{\mathcal B}_i(s_i)) \ | \ 0\leq s_i^*< s_i] \\
& = \E_{s_i^*}\left[\left.\int_{t=s_i^*}^{s_i}u_i(s_i^*,\partial\widehat{\mathcal B}_i(t))\cdot\mathcal N_i(s_i^*,\partial\widehat{\mathcal B}_i(t))
dt\ \right|\ 0\leq s_i^*< s_i\right].
\end{aligned}
\end{equation*}
Therefore, the goal of  finding $s_i$ to maximize the $P_i$'s average-case expected payoff can be written as:
\begin{equation*}
\begin{aligned}
 \underset{s_i\in(0,1)}{\arg\max}\left\{\E_{s_i^*}\left[V_i(s_i^*,\widehat{\mathcal B}_i(s_i))-\gamma_i\max\{\mathcal N_i(s_i^*,\widehat{\mathcal B}_i(s_i))-q_i,0\} \right]\right\}.
 \end{aligned}
\end{equation*}
This goal is equivalent to finding $s_i$ to minimize the  weighted sum of  OE$(s_i)$ and UE$(s_i)$ with the occurrence probabilities as the weights:
\begin{equation*}
\begin{aligned}
\underset{s_i\in(0,1)}{\arg\min} \left\{(1-F_{s_i^*}(s_i))\text{OE}(s_i)+(F_{s_i^*}(s_i)-\mathbb P(s_i^*=s_i))\text{UE}(s_i)\right\}.
\end{aligned}
\end{equation*}
By the first-order condition, the minimizer $s_i\in(0,1)$ satisfies
\begin{equation}
\label{eqn:firstorder}
\begin{aligned}
& [1-\mathbb P(s_i^*=s_i)]\E_{s_i^*}[V_i(s_i^*,\partial\widehat{\mathcal B}_i(s_i)) \ | \ s_i^*\neq s_i] \\
&\quad\quad = \gamma_i(1-F_{s_i^*}(s_i))\E_{s_i^*}[\mathcal N_i(s_i^*,\partial\widehat{\mathcal B}_i(s_i)) \ |\ s_i<s_i^*\leq 1],
\end{aligned}
\end{equation}
where $F_{s_i^*}(\cdot)$ is  the cumulative distribution function of $s_i^*\in[0,1]$.
This result proves Eq.~(\ref{eqn:quans}). 
Note that there always exists a solution to Eq.~(\ref{eqn:quans}) since when $s_i\to 0_+$, $F_{s_i^*}(s_i) \to 0, \gamma_i>u_i(s_i^*, \partial\widehat{\mathcal B}_i(s_i))$; and when $s_i\to 1_-$, $F_{s_i^*}(s_i) \to 1, V_i(s_i^*, \partial\widehat{\mathcal B}_i(s_i))>0$; and the right side of  Eq.~(\ref{eqn:quans})  is strictly decreasing in $s_i$.  
Moreover, the solution to Eq.~(\ref{eqn:quans}) is the maximizer of the thee average case expected payoff  $\E_{s_i^*}\{\mathcal U_i[\widehat{\mathcal B}_i(s_i)]\}$. This is because the  strategy corresponding to the solution of Eq.~(\ref{eqn:quans}) has a larger value of  $\E_{s_i^*}\{\mathcal U_i[\widehat{\mathcal B}_i(s_i)]\}$ than the strategy corresponding to the boundary $s_i=1$ or $s_i=0$. 

Here, we assume the calibration parameter $s_i\in(0,1)$ in the definitions of $\text{OE}(s_i)$ and $\text{UE}(s_i)$. We prove that this assumption is without less of generality by showing that if $s_i=1$, $P_i$ can pull more arms to obtain a larger expected payoff, and if $s_i=0$, $P_i$ can pull less arms to obtain a larger expected payoff.  Consider that if $s_i=1$ and $P_i$ pulls an additional arm $A$ which is not pulled currently, that is,  $A\not\in\widehat{\mathcal B}_i(1)$. Then the expected number of arms that accept agent $P_i$ is
\begin{equation*}
\mathcal N_i(s_i^*,\widehat{\mathcal B}_i(1)\cup\{A\}) = \mathcal N_i(s_i^*,\widehat{\mathcal B}_i(1)) +  \mathcal N_i(s_i^*,\{A\})
\end{equation*}
Let $\widetilde{s}_i$ satisfy $\mathcal N_i(\widetilde{s}_i,\widehat{\mathcal B}_i(1)\cup\{A\}) = q_i$. Since $\mathcal N_i(s_i^*,\widehat{\mathcal B}_i(1)\cup\{A\}) > \mathcal N_i(s_i^*,\widehat{\mathcal B}_i(1)) $, we have $\widetilde{s}_i<1$.  Let $A$ be the arm such that $\widetilde{s}_i>1-\epsilon_s$ for some sufficiently small $\epsilon_s>0$. 
Then the difference of average-case expected payoffs from pulling two arm sets $\widehat{\mathcal B}_i(1)\cup\{A\}$ and $\widehat{\mathcal B}_i(1)$  is
\begin{equation*}
\begin{aligned}
& \E_{s_i^*}\left[V_i(s_i^*,\{A\})\right] - \gamma_i\E_{s_i^*}\left.\left[q_i - \mathcal N_i(s_i^*,\widehat{\mathcal B}_i(1)\cup\{A\})\ \right| \ \widetilde{s}_i<s_i^*\leq 1\right]\\
& = \E_{s_i^*}\left[V_i(s_i^*,\{A\})\right]  - \gamma_i\E_{s_i^*}\left[\mathcal N_i(s_i^*,\{A\})\ | \ \widetilde{s}_i<s_i^*\leq 1\right] \\
& \quad\quad\quad \quad\quad\quad \quad\quad\quad \quad\quad\quad \quad\quad \quad\quad\quad + \gamma_i\E_{s_i^*}\left.\left[q_i - \mathcal N_i(s_i^*,\widehat{\mathcal B}_i(1))\ \right| \ \widetilde{s}_i<s_i^*\leq 1\right]\\
& >  \E_{s_i^*}\left[V_i(s_i^*,\{A\})\right]  - \gamma_i\E_{s_i^*}\left[\mathcal N_i(s_i^*,\{A\})\ | \ \widetilde{s}_i<s_i^*\leq 1\right] \\
& = U_i(A) \E_{s_i^*}[\mathcal N_i(s_i^*,\{A\})\ | \ 0\leq s_i^*\leq \widetilde{s}_i] - [\gamma_i - U_i(A)]\E_{s_i^*}[\mathcal N_i(s_i^*,\{A\})\ | \ \widetilde{s}_i < s_i^*\leq 1]\\
& >0,
\end{aligned}
\end{equation*}
where $U_i(A)$ is the latent utility of arm $A$ defined in Eq.~(\ref{eqn:defofutility}). The last step holds for sufficiently small $\epsilon_s>0$. Similarly, if $s_i=0$, $P_i$ can benefit by pulling less arms. Thus, the assumption that $s_i\in(0,1)$ is without less of generality.
 
If $F_{s_i^*}(\cdot)$ has discrete support, we require that  UE$(s_i)$ is at least OE$(s_i)$. By the first-order condition similar to Eq.~(\ref{eqn:firstorder}), we find  the minimal $s_i\in[0,1]$ such that
 \begin{equation}
 \label{eqn:conscaldisc}
 \begin{aligned}
& [1-\mathbb P(s_i^*=s_i)]\E_{s_i^*}[V_i(s_i^*,\partial\widehat{\mathcal B}_i(s_i))\ | \ s_i^*\neq s_i] \\
&\quad\quad \geq \gamma_i(1-F_{s_i^*}(s_i))\E_{s_i^*}[\mathcal N_i(s_i^*,\partial\widehat{\mathcal B}_i(s_i))\ |\ s_i<s_i^*\leq 1],
\end{aligned}
\end{equation}
where the search of $s_i$ starts from the maximum value in the support to the minimal value.
We note that the calibration in  Eq.~(\ref{eqn:conscaldisc}) is a \emph{conservative} counterpart as compared with the calibration such that  OE$(s_i)$ is at least UE$(s_i)$: 
\begin{equation}
 \label{eqn:noconscaldisc}
 \begin{aligned}
& [1-\mathbb P(s_i^*=s_i)]\E_{s_i^*}[V_i(s_i^*,\partial\widehat{\mathcal B}_i(s_i))\ | \ s_i^*\neq s_i] \\
&\quad\quad \leq \gamma_i(1-F_{s_i^*}(s_i))\E_{s_i^*}[\mathcal N_i(s_i^*,\partial\widehat{\mathcal B}_i(s_i))\ |\ s_i<s_i^*\leq 1].
\end{aligned}
\end{equation}
The calibration in Eq.~(\ref{eqn:conscaldisc}) is preferred to that in Eq.~(\ref{eqn:noconscaldisc}) since we want the calibration to be sensitive to the penalty $\gamma_i$.
This completes the proof.
\end{proof}

\subsection{Proof of Theorem \ref{cor:introofvalueV}: Maximin Calibration for CDM} 
\label{subsec:pfoffinoptparworst}
\begin{proof}
We use the notations $V_i(s_i^*,\widehat{\mathcal B}_i(s_i))$ and $\mathcal N_i(s_i^*,\widehat{\mathcal B}_i(s_i))$ defined in Appendix \ref{subsec:pfoffindstaave}. 
The maximum over-enrollment cost  for any  $s_i\in[0,1]$ is
\begin{equation*}
\max_{s_i^*\in[0,1]}\left\{\text{OE}(s_i) \right\} =  \gamma_i\{\mathcal N_i(1,\widehat{\mathcal B}_i(s_i)) - \mathcal N_i(1,\widehat{\mathcal B}_i(1))\} - \{V_i(1,\widehat{\mathcal B}_i(s_i)) - V_i(1,\widehat{\mathcal B}_i(1))\}.
\end{equation*}
Since $\gamma_i>\sup_{j\in\mathcal A^{T+1}}\{v_j+e_{ij}\}$, $\max_{s_i^*}\left\{\text{OE}(s_i) \right\} $ is strictly decreasing in $s_i$.  
The maximum under-enrollment cost  for any  $s_i\in[0,1]$ is
\begin{equation*}
\max_{s_i^*\in[0,1]}\left\{\text{UE}(s_i) \right\} = V_i(0,\widehat{\mathcal B}_i(0)) - V_i(0,\widehat{\mathcal B}_i(s_i)),
\end{equation*}
which is strictly increasing in $s_i$.
Hence, the goal of maximizing the minimal  total  expected payoff  $\max_{s_i}\min_{s_i^*}\{U^\text{S}_i(\widehat{\mathcal B}_i(s_i))\}$ is equivalent to minimizing the larger one between OE$(s_i)$ and UE$(s_i)$:
\begin{equation*}
\begin{aligned}
\min_{s_i\in[0,1]}\max\left\{\max_{s_i^*}\{\text{OE}(s_i)\}, \max_{s^*}\{\text{UE}(s_i)\} \right\}.
\end{aligned}
\end{equation*}
This objective amounts to finding $s_i$ such that
\begin{equation}
\label{eqn:maxsoeue}
\max_{s_i^*\in[0,1]}\{\text{OE}(s_i)\}= \max_{s_i^*\in[0,1]}\{\text{UE}(s_i)\}.
\end{equation} 
This proves Theorem \ref{cor:introofvalueV}.
Moreover, there exists a unique solution to Eq.~(\ref{eqn:maxsoeue}) since when $s_i=0$, 
$\max_{s_i^*}\left\{\text{OE}(0) \right\} >\max_{s_i^*}\left\{\text{UE}(0) \right\}=0$, and when $s_i=1$, 
$\max_{s_i^*}\left\{\text{UE}(1) \right\} > 0 = \max_{s_i^*}\left\{\text{OE}(1) \right\}$, and together with the fact that 
$\max_{s_i^*}\left\{\text{OE}(s_i) \right\} $ and $\max_{s_i^*}\left\{\text{UE}(s_i) \right\} $ are monotonic continuous functions of $s_i$.

If $F_{s_i^*}(\cdot)$ has discrete support, we requires that the maximal UE$(s_i)$ is at least the maximal OE$(s_i)$. Hence, we need to change the goal in Eq.~(\ref{eqn:maxsoeue}) to finding the minimal $s_i\in[0,1]$ such that $ \max_{s_i^*}\{\text{UE}(s_i)\}\geq \max_{s_i^*}\{\text{OE}(s_i)\}$.
This completes the proof.
\end{proof}

\subsection{Proof of Theorem \ref{cor:cdmagentincentive}: Incentives of Agents} 
\label{subsec:pfofincentives}

\begin{proof}
First, by Theorem \ref{thm:cutoffaystate}, the CDM takes the cutoff strategy with respect to the arms' latent utilities  given by Eq.~(\ref{eqn:defofutility}). Since the agents’ true preferences are determined by the arms’ latent utilities, the agents would act according to their true preferences under the CDM.

Second, Theorem \ref{thm:optestpi} proves the consistency of the estimated acceptance probability using historical data. Theorems \ref{thm:optimalstates} and  \ref{cor:introofvalueV} show that the CDM maximizes the agent's expected payoff, given the population acceptance probability, where the expected payoff  is measured in either average-case or worst-case given the uncertain true state.  Thus, it is a dominant strategy for each agent to act according to the CDM. 

Combining these two observations, we conclude that as $T\to\infty$, the CDM is a procedure that gives agents the incentives to act according to true preferences. 
\end{proof}

\subsection{Proof of Theorem \ref{thm:stabcutoff}: Stability of CDM} 
\label{subsec:pfofstacutff}
\begin{proof}
Suppose that an agent-arm pair $(P_i,A_j)$ forms a blocking pair. Then it implies one of the following two cases:
\begin{itemize}
\item[(I).] $P_i$ prefers $A_j$ to some of its matched arms. 
\item[(II).]$P_i$ has unfilled quota and $A_j$ is unmatched. 
\end{itemize}

For case (I), $P_i$ must have pulled $A_j$ according to the cutoff strategy of CDM  by Theorem \ref{thm:cutoffaystate}.
However, $P_i$ must have been  subsequently rejected by $A_j$ in favor of some agent that $A_j$ liked better. 
Hence, $A_j$ must prefer its currently matched agent to $P_i$ and there is no instability. 

For case (II), $P_i$ did not pull $A_j$ since otherwise, $A_j$ would have accepted $P_i$. 
Note that Theorem \ref{thm:optestpi} proves the consistency of acceptance probability estimate using historical data.
The individual rationality in Eq.~(\ref{eqn:accepcond}) implies that  $P_i$ would find $A_j$ unacceptable in the decentralized market, given the population acceptance probability. Thus, there is  no instability.  

Combining these two cases, we conclude that as $T\to\infty$, CDM yields a stable matching in decentralized markets. 
\end{proof}

\subsection{Proof of  Theorem \ref{prof:faircdm}: Fairness of CDM}
\label{sec:profoffaircdm}
\begin{proof}
By Theorem \ref{cor:cdmagentincentive}, the CDM gives agents the incentives to act on their true preferences, as $T\to\infty$. Hence, each agent pulls arms according to its true preference for arms. If an arm $A_j$ prefers an agent $P_{i'}$ to another agent $P_i$ that pulls $A_j$, then all arms pulled by $P_{i'}$ must rank above $A_j$ according to the true preference of $P_{i'}$.
By definition, we conclude that the matching process according to CDM is fair.
\end{proof}

\subsection{Proof of Theorem \ref{thm:unfairnessoracle}: Unfairness of the Oracle Set}
\label{sec:appproofofunfair}
\begin{proof}
Denote by $e_i^*(v)$ the cutoff curve corresponding to the oracle arm set $\mathcal B_i^*$ in  Eq.~(\ref{eqn:oracleset}):
\begin{equation*}
\begin{aligned}
e_i^*(v_j)  = \min\left\{\max\left\{ \gamma_i\cdot\mathbb P(s_i^*\in O_{\mathcal B^*_i})\frac{\E_{s_i^*}[\pi_i(s_i^*,v_j)\ |\ s_i^*\in O_{\mathcal B^*_i}]}{\E_{s_i^*}[\pi_i(s_i^*,v_j)]} - v_j, 0\right\}, 1\right\},\quad\forall j.
\end{aligned}
\end{equation*}
The average-case expected utility of an arm from the cutoff curve is
\begin{equation*}
\mathcal U_i(v_j,e_i^*(v_j)) = (v_j+e^*_{i}(v_j))\cdot\E_{s_i^*}[\pi_i(s_i^*,v_j)] - \gamma_i\cdot\mathbb P(s_i^*\in O_{\mathcal B^*_i})\E_{s_i^*}[\pi_i(s_i^*,v)\ |\ s_i^*\in O_{\mathcal B^*_i}].
\end{equation*}
Hence, for $e_i^*(v_j)\in(0,1)$,
\begin{equation*}
    \mathcal U_i(v_j,e_i^*(v_j)) = 0.
\end{equation*}

Since the acceptance probability $\pi_i(s_i,v)$ is  increasing in $s_i$, 
$\E_{s_i^*}[\pi_i(s_i^*,v)\ |\ s_i^*\in O_{\mathcal B^*_i}] > \E_{s_i^*}[\pi_i(s_i^*,v)]$. Thus, for $e_i^*(v_j)\in(0,1)$,
\begin{equation}
\label{eqn:defofuistarvei}
\begin{aligned}
v_j+e^*_{i}(v_j) & = \gamma_i\cdot\mathbb P(s_i^*\in O_{\mathcal B^*_i})\frac{\E_{s_i^*}[\pi_i(s_i^*,v_j)\ |\ s_i^*\in O_{\mathcal B^*_i}]}{\E_{s_i^*}[\pi_i(s_i^*,v_j)]} \\
& > \gamma_i\cdot\mathbb P(s_i^*\in O_{\mathcal B^*_i}).
\end{aligned}
\end{equation}
By Eq.~(\ref{eqn:changeparpisvj}),  we can derive that
\begin{equation*}
\E_{s_i^*}\left[\frac{\partial\pi_i(s_i^*,v_j)}{\partial v}\right] < \E_{s_i^*}\left[\left.\frac{\partial\pi_i(s_i^*,v_j)}{\partial v}\ \right|\ s_i^*\in O_{\mathcal B^*_i}\right]<0,\quad\forall v_j\in(v',v'').
\end{equation*}
This inequality together with Eq.~(\ref{eqn:defofuistarvei}) yield that for $v_j\in(v',v'')$ and $e_i^*(v_j)\in(0,1)$,
\begin{equation*}
\begin{aligned}
\left.\frac{\partial \mathcal U_i(v,e_i^*(v))}{\partial v}\right\vert_{v=v_j} & = \E_{s_i^*}[\pi_i(s_i^*,v_j)] +  (v_j+e^*_{i}(v_j))\E_{s_i^*}\left[\frac{\partial\pi_i(s_i^*,v_j)}{\partial v}\right]\\
&\quad\quad\quad\quad - \gamma_i\cdot\mathbb P(s_i^*\in O_{\mathcal B^*_i})\E_{s_i^*}\left[\left.\frac{\partial\pi_i(s_i^*,v_j)}{\partial v}\ \right|\ s_i^*\in O_{\mathcal B^*_i}\right]\\
&<\E_{s_i^*}[\pi_i(s_i^*,v_j)].
\end{aligned}
\end{equation*}
Thus, by the implicit function theorem,
\begin{equation*}
\begin{aligned}
\left|\frac{de^*_i(v_j)}{dv}\right|  = \left|\frac{\partial \mathcal U_i(v_j,e_i^*(v_j))/\partial v}{\partial \mathcal U_i(v_j,e_i^*(v_j))/\partial e_i^*}\right|<\frac{\E_{s_i^*}[\pi_i(s_i^*,v_j)]}{\E_{s_i^*}[\pi_i(s_i^*,v_j)]}=1,\quad\forall v_j\in(v',v''),
\end{aligned}
\end{equation*}
Therefore, we can  find two arm sets $\mathcal B_i^{(1)},\mathcal B_i^{(2)}\subseteq (v',v'')\times [0,1]$ and a constant $c_0$ such that, for all $(v_j,e_{ij})\in \mathcal B_i^{(1)}, e_{ij}>e^*_i(v_j)$ and $v_j+e_{ij}<c_0$, and for all $(v_j,e_{ij})\in \mathcal B_i^{(2)}, e_{ij}<e^*_i(v_j)$ and $v_j+e_{ij}>c_0$. Hence, the arms in $\mathcal B_i^{(2)}$ have justified envy toward arms in $\mathcal B_i^{(1)}$. We refer to  Figure \ref{fig:faircdm} for an illustration.
By definition of fairness, the strategy corresponding to the oracle set  in  Eq.~(\ref{eqn:oracleset}) is unfair.
\end{proof}

\end{document}